\documentclass[sigconf]{acmart}
\AtBeginDocument{%
  \providecommand\BibTeX{{%
    \normalfont B\kern-0.5em{\scshape i\kern-0.25em b}\kern-0.8em\TeX}}}

\setcopyright{acmcopyright}
\copyrightyear{2018}
\acmYear{2018}
\acmDOI{XXXXXXX.XXXXXXX}

\acmConference[Conference acronym 'XX]{Make sure to enter the correct
  conference title from your rights confirmation emai}{June 03--05,
  2018}{Woodstock, NY}
\acmPrice{15.00}
\acmISBN{978-1-4503-XXXX-X/18/06}

\usepackage{graphicx}
\usepackage{xcolor}    
\usepackage{amsmath,array}
\usepackage{amsthm}
\usepackage{colortbl}
\usepackage{threeparttable}
\usepackage{algorithm}
\usepackage{multirow}
\usepackage{amsfonts}
\usepackage{caption}
\usepackage{bbold}
\usepackage{color}
\usepackage{mathtools}
\usepackage{makecell}
\usepackage{booktabs}
\usepackage{enumitem}
\usepackage{makecell}

\usepackage[noend]{algpseudocode}
\usepackage{algorithmicx}

\usepackage{amsfonts,amssymb}
\newtheorem{remark}{Remark}

\newtheorem{thm}{Theorem}
\newenvironment{thmbis}[1]
 {\renewcommand{\thethm}{\ref{#1}}%
 \addtocounter{thm}{-1}%
 \begin{thm}}
 {\end{thm}}

\begin{document}

\title[DREditor: An Time-efficient Approach for Building Domain-specific DR Model]{DREditor: An Time-efficient Approach for Building a Domain-specific Dense Retrieval Model}

\author{
    Chen Huang$\dag$, Duanyu Feng$\dag$, Wenqiang Lei, Jiancheng Lv
}

\email{huangc.scu@gmail.com, fengduanyu@stu.scu.edu.cn}
\affiliation{
    %Afiliations
    \institution{College of Computer Science, Sichuan University\\Engineering Research Center of Machine Learning and Industry Intelligence, Ministry of Education}
    \city{Chengdu}
    \country{China}
}

\renewcommand{\shortauthors}{Trovato and Tobin, et al.}

\begin{abstract}
Deploying dense retrieval models \textit{efficiently} is becoming increasingly important across various industries. This is especially true for enterprise search services, where customizing search engines to meet the time demands of different enterprises in different domains is crucial.
Motivated by this, we develop a time-efficient approach called DREditor to edit the matching rule of an off-the-shelf dense retrieval model to suit a specific domain. 
This is achieved by directly calibrating the output embeddings of the model using an efficient and effective linear mapping. This mapping is powered by an edit operator that is obtained by solving a specially constructed least squares problem.
Compared to implicit rule modification via long-time finetuning, our experimental results show that DREditor provides significant advantages on different domain-specific datasets, dataset sources, retrieval models, and computing devices. It consistently enhances time efficiency by 100–300 times while maintaining comparable or even superior retrieval performance.
In a broader context, we take the first step to introduce a novel embedding calibration approach for the retrieval task, filling the technical blank in the current field of embedding calibration. This approach also paves the way for building domain-specific dense retrieval models efficiently and inexpensively. Codes are available at \url{https://github.com/huangzichun/DREditor}.
\end{abstract}





\received{20 February 2007}
\received[revised]{12 March 2009}
\received[accepted]{5 June 2009}

\maketitle
\def\thefootnote{$\dag$}\footnotetext{Both authors contributed equally to this study.}\def\thefootnote{\arabic{footnote}}

\textbf{Relevance to The Web Conference}. 
This work investigates the challenge of time-efficiently building a domain-specific dense retrieval model. As a result, it meets the requirements of \textit{Search} Track and is closely related to the topic of "\textit{Web search models and ranking}" and "\textit{Web question answering}".

\section{Introduction}
% draft v2
%In order to stay competitive in today's business world, small and medium-sized enterprises (SMEs) need to upgrade their technologies as AI continues to rapidly develop \cite{bhalerao2022study}. However, as NLP models continue to increase in size, the amount of training time required to customize these models to fit a company's specific needs can be a significant challenge for these businesses, particularly when they have limited computational resources \cite{chen2021bird, fan2019scheduling}. \textbf{We focus on these time-efficient model customization scenarios, with particular focus on the domain knowledge retrieval}, which is the foundation of Cloud Customer Service (CCS). 
%Businesses from various industries can upload the domain-aware documents (i.e., operating manual) to the CCS platform and aim to quickly configure an AI service to engage intelligently and proactively with customers by answering their questions, thereby resolving user demands in a timely manner.

The efficient deployment of dense retrieval models is becoming increasingly crucial for a wide range of industries. This need is particularly prominent in research and industry areas such as enterprise search (ES), where a single search service is commonly used to support multiple enterprises operating in diverse domains \cite{rose2012cloudsearch, tran2019domain, bendersky2022search}.
%In the world of enterprise search (ES), a single search service is often utilized to power multiple enterprises across various domains \cite{rose2012cloudsearch, tran2019domain, bendersky2022search}. 
One example of ES is the Cloud Customer Service\footnote{Such as Amazon: \url{https://aws.amazon.com/} and Google Cloud: \url{https://cloud.google.com/}} (CCS), where enterprises upload business documents (i.e., operating manual of a product) to the CCS platform to obtain a personalized search engine that can assist customers by answering their questions and addressing their needs. 
Due to the scalability requirements of various enterprises, the success of ES providers hinges on their ability to deliver \textit{time-efficient} searching customization. 
If they fail to do so, it could lead to delays in resolving enterprise needs and ultimately result in a poor experience and loss of business.

The impressive text representation capabilities of pre-trained language models have enabled the development of search engines using dense retrieval (DR) models \cite{tran2019domain}, which establish a \textit{matching rule} that ensures the embedding of a question corresponds to the embedding of its relevant answer\footnote{We use the terms "query" and "question" interchangeably, as well as the terms "corpus" and "answer".} \cite{karpukhin-etal-2020-dense, xiongapproximate}. 
Despite its superiority, a single DR model cannot be applied to multiple enterprises as it does not generalize well across domains. 
This has led to studies on customizing domain-specific DR models \cite{thakur2021beir, yu-etal-2022-coco} to meet the requirements of ES. More importantly, it is crucial for the customization process to be \textit{time-efficient}, as any delays could negatively impact the platform usage experience for all enterprises.
The need for efficient search customization is becoming more pressing as DR models increase in size \cite{bhalerao2022study}. This poses a challenge for small and medium-sized ES providers with limited computational resources \cite{chen2021bird, fan2019scheduling}, as customizing a large DR model to fit an enterprise's domain-specific needs is becoming more time-consuming. This could make these providers less competitive in today's business world. Therefore, \textbf{we aim to study the challenge of efficiently editing the matching rule of a DR model so that it can adapt to a specific domain.}

\begin{figure}[t]
    \centering
    \includegraphics[width=0.45\textwidth]{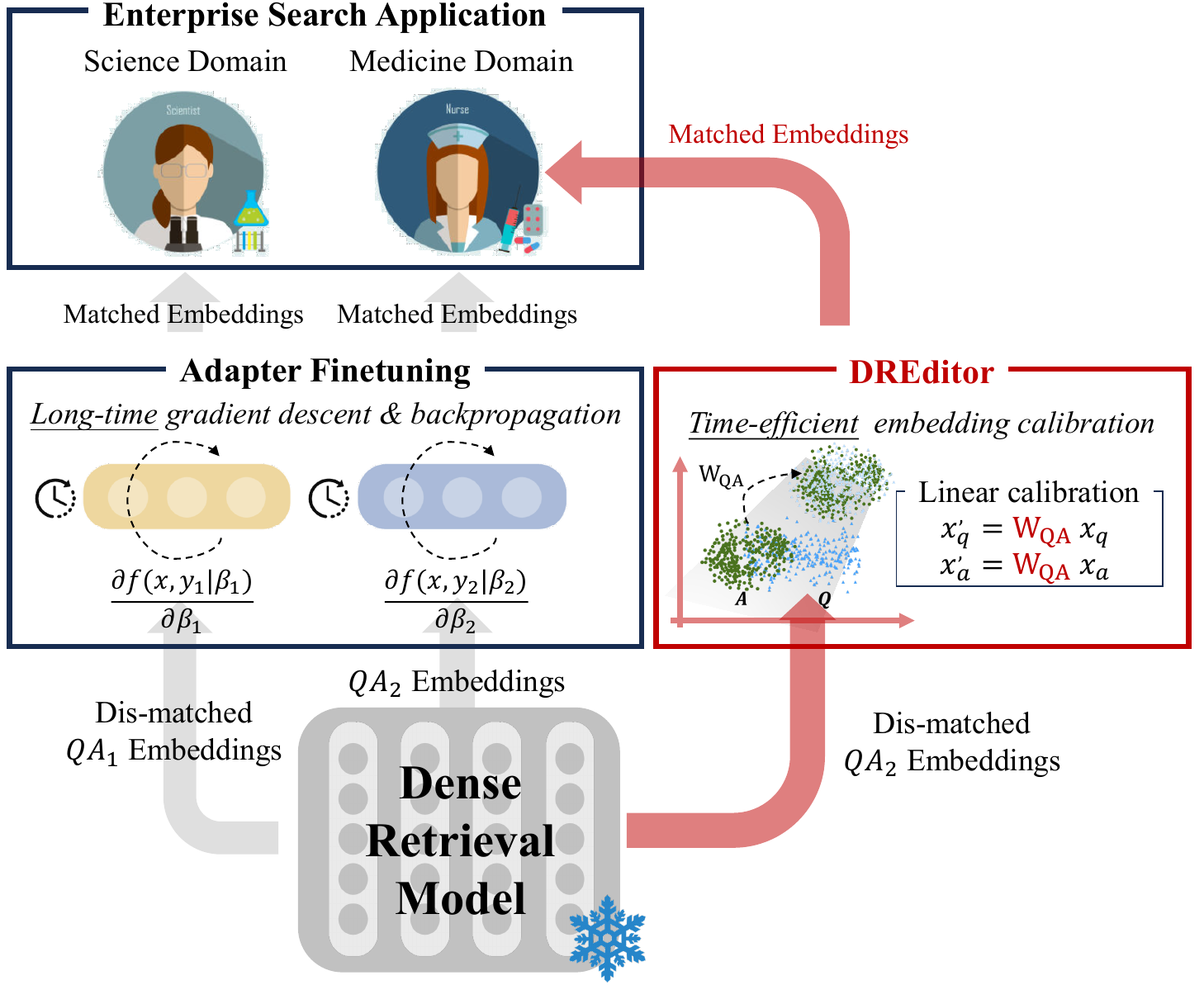}
    \vspace{-1mm}
    \caption{\textcolor{gray}{Fine-tuning} a DR model for each enterprise is time-consuming, while \textcolor{red}{DREditor} is time-efficient due to efficient embedding calibration without long-time parameter tuning.
    }
    \label{fig:my_labdel3333}
    \vspace{-3mm}
\end{figure}

One possible solution, as shown in Figure \ref{fig:my_labdel3333}(Left), is to use \textit{adapter fine-tuning}, which involves tuning only a small part of the original parameters of the model and storing the tuned parameters for each enterprise. 
However, this approach requires long-time iterative gradient descent and backpropagation optimization, which can take at least a day to finish even using high-priced GPUs \cite{pal2023parameter, tam2022parameter, ruckle2020multicqa}. Not to mention the time required to optimize an over-parameterized large model \cite{arabzadeh2021predicting, ren2022thorough}.
Such time consumption is intolerable for time-starved enterprise search services, especially for small and medium-sized ES service providers with limited computational resources.
As a result, adapter fine-tuning may not be the best approach for real-world ES applications due to its long-time tuning process \cite{nityasya2020costs, arabzadeh2021predicting}. 
%\textbf{How to efficiently edit the matching rule of a DR model so that it can adapt to a specific domain without taking up too much time is an urgent problem}.

% 

To this end, we are pioneering the development of a time-efficient method, called DREditor, to edit the matching rule of an off-the-shelf DR model in order to adapt it to a specific domain. Instead of using long-time parameter optimization to edit the rule, DREditor directly calibrates the output embeddings of the model using an efficient linear mapping technique. This is applied as an efficient and post-processing step to ensure that the question embeddings are accurately aligned with their corresponding corpus embeddings.
More specifically, the linear mapping is customized and stored separately for each enterprise, and it is parameterized by an edit operator $W_{QA}$. To maintain time efficiency, we calculate $W_{QA}$ in a closed form by solving a specially constructed least squares problem. Theoretically, solving such a closed form is significantly faster than the gradient descent method, which requires multiple steps to explore the parameter space, and it helps specify the domain-specific semantic associations that are not accounted for in the current matching rule of the DR model. As a result, DREditor boosts the generalization capacity of a DR model on a specific domain without the need for long-time iterative gradient descent and backpropagation methods.
In a broader context, we take the first step to introduce a novel embedding calibration technique for the retrieval task, filling the technical blank in the current field of embedding calibration \cite{faruqui2015retrofitting, mrkvsic2017semantic, glavavs2018explicit, zervakis2021refining}. This also paves the way for building domain-specific dense retrieval models efficiently and inexpensively.
% Such advantages even hold when dealing with zero-shot domain-specific DR scenario (ZeroDR), where the DR model is tuned on the source domain \cite{thakur2021beir, neelakantan2022text, izacard2022unsupervised}.

% Our method is related to recent studies on embedding calibration, but there are currently no calibration methods specifically designed for dense retrieval. Therefore, we do not compare DREditor to existing calibration methods.
% Instead, w
Experimentally, we first compare the DREditor with the adapter finetuning to demonstrate our time efficiency. In this scenario, both methods can use the target-domain data (train data) to edit or tune the DR model. To verify the versatility of DREditor, we also test DREditor on the zero-shot domain-specific DR scenario (ZeroDR)\footnote{ZeroDR is an emerging research field that can be utilized to deal with situations where the ES providers are unable to collect data specific to a particular enterprise to customize their search services.}, 
where the DR model is tuned on the source-domain data (some external data) to ensure its effectiveness \cite{thakur2021beir, neelakantan2022text, izacard2022unsupervised}. 
%which is an emergency situation in the search area. In this scenario, the DR model was tuned on the source domain to ensure its effectiveness \cite{thakur2021beir, neelakantan2022text, izacard2022unsupervised}. 
DREditor has significant advantages on different domain-specific datasets (i.e., finance, science, and bio-medicine), DR models, and computing devices (GPU and CPU).
In particular, the results indicate that DREditor successfully adapts an off-the-shelf DR model to a specific domain and consistently improves retrieval effectiveness compared to the vanilla DR model.
Furthermore, DREditor significantly enhances time efficiency by 100–300 times compared to the long-time adapter finetuning while maintaining comparable or even superior retrieval performance when dealing with ZeroDR.
To provide a comprehensive understanding of DREditor, we conducted intrinsic experiments to visualize the calibrated embedding distribution.
The results show that DREditor translates the output embeddings of a DR model to a specialized space with a decreased mean discrepancy between the embeddings of the questions and answers. 
Even when dealing with ZeroDR, utilizing source-domain data to efficiently calibrate DR embedding via DREditor is still a promising solution to alleviate the embedding distribution shift between the source-domain and target-domain data.
In this paper, we recognize our work as an important proof-of-concept for building a domain-specific DR model time-efficiently and inexpensively. To sum up, we claim the following contributions.

% Such advantages even hold when dealing with zero-shot domain-specific DR scenario (ZeroDR), where the DR model is tuned on the source domain \cite{thakur2021beir, neelakantan2022text, izacard2022unsupervised}. 
% In particular, the results show that DREditor successfully adapts an off-the-shelf DR model to a specific domain and gives a consistent improvement in retrieval effectiveness over the raw DR model. Compared to conventional adapter finetuned models that involve gradient propagation on GPUs, our method gives substantial improvements in time efficiency, 100-300 times faster without significantly compromising retrieval performance.

\begin{itemize}
    \item We call attention to the challenge of efficiently deploying dense retrieval (DR) models for various industries' applications, including enterprise search (ES), to meet their search customization requirements.
    % We call attention to the challenge of efficiently building domain-specific DR models to meet the requirements of ES providers on the searching customization.
    \item For the first time, we propose a time-efficient approach called DREditor to edit the matching rule of a DR model to suit a specific domain by embedding calibration, which efficiently specifies the domain-specific semantic associations that are not accounted for in the given DR model via two matrix multiplications, ensuring our efficiency. 
    \item In DREditor, we propose a novel embedding calibration method for DR models. We lay its mathematical foundation by solving a specially constructed least squares problem in a closed form. 
    % For the first time, we propose a time-efficient approach based on mathematical derivation to directly edit the matching rule of a DR model in a post-hoc manner by embedding calibration. This is achieved by solving a specially constructed least squares problem in a closed form without long-time gradient descent and backpropagation.
    \item We verify our efficiency and effectiveness with empirical studies and give a successful implementation of time-efficiently building domain-specific DR. We set a landmark for future studies of efficient editing DR model.
\end{itemize}

\section{Related Work}
Our focus is on efficiently editing a DR model to adapt to a specific domain. We conducted a literature review on domain-specific dense retrieval. Since our method follows the research on embedding calibration and model editing, we also discuss our differences.
%\footnote{Due to limited space, we also discuss our difference to the embedding retrofitting in Appendix \ref{seeherebaby}.}. %修改

% 现在对了吗
\begin{figure*}[t]
    \centering
    \includegraphics[width=0.9\textwidth]{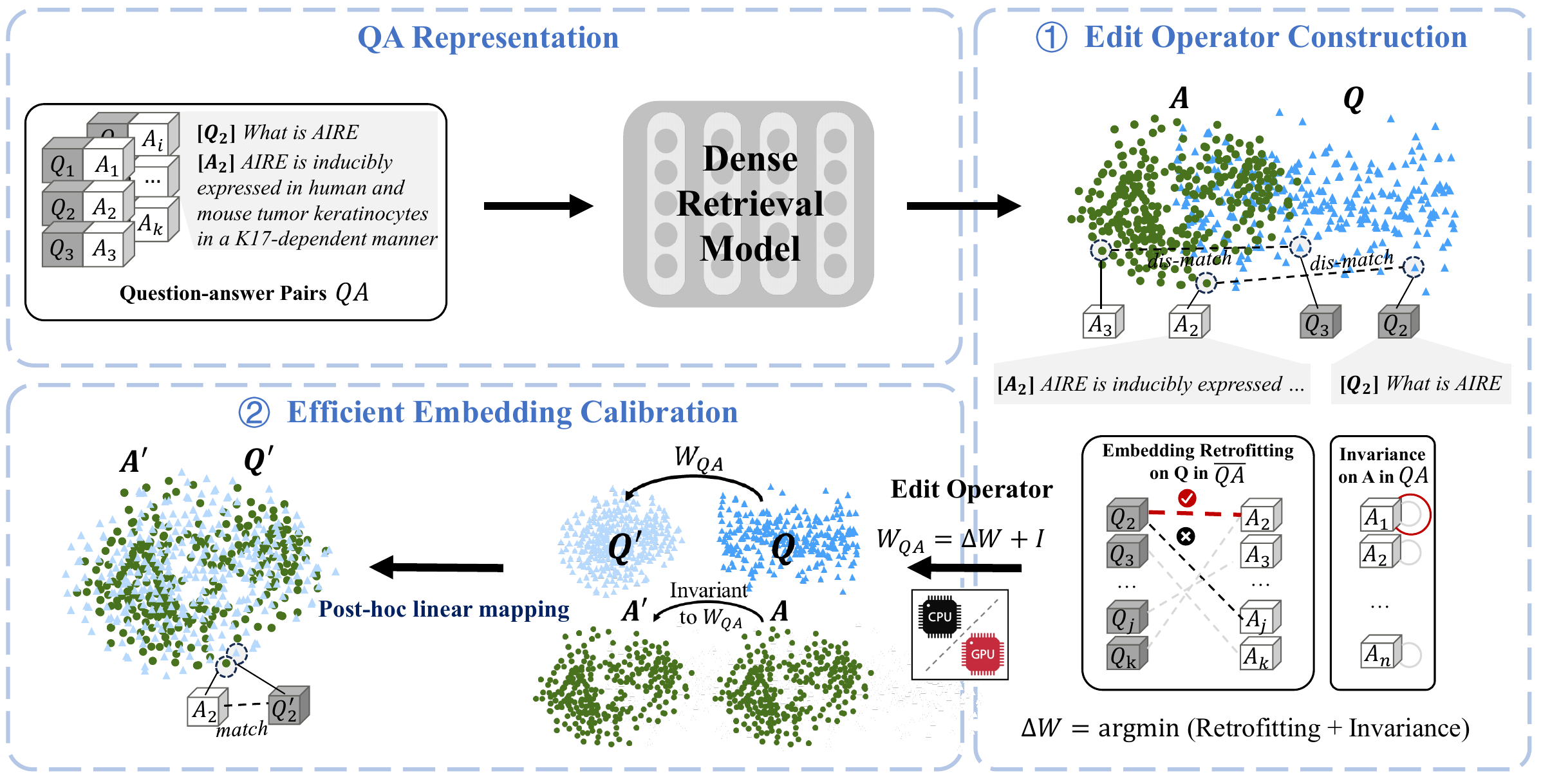}
    \vspace{-2mm}
    \caption{Pipline of DREditor, which aims to find a linear mapping, powered by $W_{QA}$, to correct the mismatched pairs. Here, $W_{QA}$ is derived by solving a specifically constructed least squares problem.}
    \label{fig:my_labdel}
    \vspace{-3mm}
\end{figure*}

\textbf{Domain-Specific DR and ZeroDR.} 
Pre-training over-paramet-erized PLMs to obtain a better domain generalization on retrieval tasks is known to be time-consuming \cite{10.1145/3397271.3401075, santhanam2022colbertv2}.
The canonical approach to adapting a DR to a specific domain is to collect training data and use it to train the DR model, such as by finetuning SBERT \cite{reimers-2019-sentence-bert}, DPR \cite{karpukhin-etal-2020-dense}, and ANCE \cite{xiongapproximate}. 
When it's not possible to collect in-domain training data, zero-shot domain-specific DR (ZeroDR) \cite{thakur2021beir, neelakantan2022text, izacard2022unsupervised} has emerged as a promising solution, which also needs to finetune source-domain data like the background knowledge base \cite{zhu2022enhancing, seonwoo2022virtual, rocktaschel2015injecting}.
However, existing methods necessitate long-time iterative gradient backpropagation \cite{arabzadeh2021predicting, ren2022thorough} for finetune, which can take at least a day to finish even using high-priced GPUs \cite{pal2023parameter, tam2022parameter, ruckle2020multicqa}. 
Not to mention the time required to optimize an over-parameterized large model \cite{arabzadeh2021predicting, ren2022thorough}.
To improve the training efficiency, parameter-efficient finetuning methods like LoRA \cite{hu2021lora} are proposed. However, recent evidence suggests that these methods may need more training time \cite{guo2020parameter, houlsby2019parameter, fu2023effectiveness} as they prioritize reducing the number of parameters and improving device usage. 
As a result, long model training times limit the use of domain-specific retrieval models in real-world enterprise search applications \cite{nityasya2020costs, arabzadeh2021predicting}.
This motivates us to push the limit of building domain-specific DR efficiently to meet the requirements of time-demanding ES services.

%%%knowledge editor 
\textbf{Embedding Calibration.} 
% Embedding calibration is an emerging area in the communities of NLP \cite{lengerich2018retrofitting, dong-etal-2022-calibrating} and image processing \cite{chen2022uncertainty, fisch2022calibrated, santurkar2021editing, bau2020rewriting}. For NLP, it aims at incorporating semantic knowledge into word embeddings and translating them into the specialized space that meets the semantic constraints \cite{glavavs2018explicit, shah2020retrofitting}, such as synonyms-antonyms \cite{mrksic-etal-2016-counter}, hypernyms-hyponyms \cite{faruqui2015retrofitting}, and affective lexicons constraints \cite{shah2022affective, seyeditabari2019emotional}. 
Embedding calibration is an emerging area \cite{lengerich2018retrofitting, dong-etal-2022-calibrating} that incorporates semantic knowledge into word embeddings and translates them into specialized spaces that meet semantic constraints \cite{glavavs2018explicit, shah2020retrofitting}.
These constraints include synonyms-antonyms \cite{mrksic-etal-2016-counter}, hypernyms-hyponyms \cite{faruqui2015retrofitting}, and affective lexicons constraints \cite{shah2022affective, seyeditabari2019emotional}.
While existing methods have been validated for their effectiveness in multiple downstream tasks, they are limited to classification \cite{lengerich2018retrofitting} or generation \cite{wald2021calibration} tasks. Calibrating the embeddings of a DR model for domain-specific retrieval tasks remains unexplored. 
Therefore, before our DREditor, there was a lack of quick calibration methods specifically designed for dense retrieval that could be compared or referenced.
For the first time, we propose to fill this blank and demonstrate our effectiveness experimentally and theoretically. Moreover, unlike most existing methods that focus on non-contextual embeddings and only one specified data format, our DREditor is designed to calibrate the contextual embeddings and has more data-format scalability. For more information on how DREditor works, please refer to Section \ref{7yuqw34ukryqeawib}.

\textbf{Model Editing.} 
Model editing is an emerging field that seeks to facilitate rapid modifications to the behavior of a model \cite{mitchell2022memory}. 
The most commonly chosen methods modify the model behavior by involving the external domain information that is absent from the model.
Usually, these domain information takes the form of codebooks \cite{hartvigsen2022aging, chenglanguage} and key-value memories \cite{yao2022kformer, chenglanguage, dong-etal-2022-calibrating, cheng2023editing}. 
However, not only can existing methods not be used for DR tasks, but also those methods still require the long-time training process, leaving efficient model modification a goal striving for all along.
Calling attention to the efficient model editing, \citep{santurkar2021editing, bau2020rewriting} propose making the rank-one updates to the model parameter without iterative optimization. However, they could only edit an image classifier \cite{santurkar2021editing} and a generative image model \cite{bau2020rewriting}.
Unlike these methods, we take the first steps to propose a method to edit a DR model for text matching tasks. Additionally, our method is time-efficient without iterative gradient descent and backpropagation.

\section{DREditor for efficiently Editing DR}
To distinguish our approach from adapter-based parameter tuning, we first present the task formalization and explain how we edit the matching rule of a DR model using embedding calibration in Section  \ref{tawefn2093r}. This involves using an edit operator called $W_{QA}$ to perform a linear mapping. We then demonstrate how to efficiently construct this edit operator in Section \ref{7yuqw34ukryqeawib}. Finally, we elaborate on the efficient rule calibration in a post-hoc manner and even use it in the zero-shot DR scenario in Section \ref{thisiswhatisit}.

% In the following, we first show the task formalization of DREditor which involves a edit operator $W_{QA}$ to perform the linear embeddings mapping, and then show how to construct such the edit operator $W_{QA}$. We finally move on to showing how to use $W_{QA}$ for efficient rule calibration in a post-hoc manner and even use it in the zero-shot DR scenario.

\subsection{Task Formalization}
\label{tawefn2093r}
The traditional approach to modifying the matching rule of a DR model involves gathering domain-specific data and finetuning the model via gradient descent and backpropagation, which is time-consuming. DREditor, on the other hand, directly edits the matching rule by calibrating the output embeddings of the DR. 
Formally, as shown in Figure \ref{fig:my_labdel}, 
DREditor uses a dataset of question-answer pairs $QA$, where each pair $(q, a) \in QA$ consists of a domain-related question $q$ and its corresponding answer $a$. Additionally, it uses a dataset of mismatched pairs $\overline{QA} \subset QA$, where the semantic association between the question and answer in the embedding space is low\footnote{Measured by whether the DR retrieves the corresponding answer from the candidates based on the input question.}.
\textbf{DREditor aims to find a linear mapping, powered by $W_{QA}$, to correct the mismatched pairs so that $W_{QA}x_q=x_a$}, implying that domain-specific question embedding $x_q$ is accurately matched to its relevant answer embedding $x_a$.
We refer to $W_{QA}$ as an edit operator because it edits the embedding $x_q$ to modify the associations between the domain-related question and its answer. This feature allows for the post-hoc editing of the matching rule. It's worth noting that the edit operator is designed to be adaptable to DR models of various architectures, not limited to just the two-tower architecture like DPR \cite{karpukhin-etal-2020-dense}.

\subsection{Edit Operator Construction}
\label{7yuqw34ukryqeawib}
DREditor utilizes a linear mapping, powered by the edit operator $W_{QA}$, to directly calibrate the output embeddings of the model. This mapping is designed to be efficient and effective. To obtain the edit operator, we solve a specifically constructed least squares problem. In this subsection, we show how the edit operator is constructed and how it ensures time efficiency.

\textbf{Edit Operator Formalization}. Given mismatched QA pairs $\overline{QA}$, we are motivated to develop the edit operation to enhance the semantic association between questions and corresponding answers in $\overline{QA}$. Intuitively, one could minimize the following squared error of question-answer associations. 
\begin{equation}
    W_{QA} = \arg \min_W \sum_{(x_q, x_a) \in \overline{QA}} \| Wx_q - x_a \|_2^2.
    \label{eq1}
\end{equation}
Although efficiently solving the objective mentioned above is feasible, modifying $x_q$ without affecting $x_a$ poses a challenge when both embeddings come from the same DR encoder. To address this issue, we need to retrofit $x_q$ while keeping $x_a$ unchanged.
In essence, we expect $W_{QA}$ to act as an identity matrix $I$ for $x_a$, while maximizing the association between $x_q$ and $x_a$. Mathematically, this can be represented as $x_a = W_{QA}x_a$ and $x_a = W_{QA}x_q$.
To facilitate optimization, we assume that $W_{QA} = I + \Delta W$, where $\Delta W$ captures the associations in the question-answer pairs $\overline{QA}$ while preserving the representation of $x_a$ as much as possible. Expanding on this idea, in addition to the retrofitting term in Equation (\ref{eq1}), we explicitly include an invariance constraint on $x_a$. This leads to the following weighted least square optimization problem.
%\begin{small} 
\begin{equation}
\begin{split}
    \Delta W = & \arg \min_{\overline{W}} \sum_{(x_q, x_a) \in \overline{QA}} \| (I + \overline{W})x_q - x_a \|_2^2 \\
    & +  \beta \sum_{x_a \in QA} \|  (I + \overline{W})x_a - x_a \|_2^2.
\end{split}
\label{eq3}
\end{equation}
%\end{small}
%% @colfeng, why QA, not \overline{QA}.
Here, the first term in Eq (\ref{eq3}) edits the association between questions and answers,
while the second term preserves the representation of $x_a\in QA$ without any changes. The $\beta>0$ is a hyper-parameter that allows us to control the importance of maintaining the original answer embedding $x_a$ when performing the editing operation.

\textbf{Solving Edit Operator Efficiently}. 
To compute $\Delta W$, the first term in Equation (\ref{eq3}) necessitates the presence of $\overline{QA}$. This means that we have to go through all the QA pairs (by performing forward computation) to identify all the mismatched QA pairs. This process can be quite time-consuming, which can reduce the overall efficiency of our approach. To address this issue, we present a theorem that allows us to optimize $\Delta W$ efficiently without involving $\overline{QA}$ in the optimization process.
\begin{theorem}
    The optimization problem (\ref{eq3}) could be rephrased as the following form that involves only $QA$, where $\hat\beta$ is a hyper-parameter associated with $\beta$. When $\beta \gg 1$, $\hat\beta \approx \beta$.
    % can introduce weights to replace $\overline{QA}$ and transform into the following form:
    %\begin{small}
    \begin{equation}
        \begin{split}
    \Delta W = & \arg \min_{\overline{W}} \sum_{(x_q, x_a) \in QA} \| (I + \overline{W})x_q - x_a \|_2^2 \\
    & +  \hat\beta \sum_{x_a \in QA} \|  (I + \overline{W})x_a - x_a \|_2^2.
        \end{split}
    \label{finaleq}
    \end{equation}
    %\end{small}
    \label{theory}
\end{theorem}
\begin{proof}
    The theorem mentioned above can be demonstrated through the application of certain decomposition and transformation techniques. Refer to Appendix \ref{proof} for detailed proof.
\end{proof}

Theorem \ref{theory} provides a way to optimize $\Delta W$ in an efficient manner, without the need to construct the mismatched pairs $\overline{QA}$.
Note that the optimization problem (\ref{finaleq}) is a well-studied least-squares problem \cite{bau2020rewriting}, which can be solved efficiently using closed form.
\begin{theorem}
    $\Delta W= (X_a X_q^T - X_q X_q^T)(\hat\beta^2 X_a {X_a}^T + X_q X_q^T)^{-1}$ is the minimizer of the problem (\ref{finaleq}). 
    \label{7yujmpol}
\end{theorem}
\begin{proof}
Our proof starts with the following matrix form of the problem (\ref{finaleq}), where we use the notation $[A|B]$ to be the concatenation operator between matrices or vectors. $X_q$ and $X_a$ are matrix form of the embedding of questions $[x^1_q|x^2_q|\cdots|x^n_q]$ and answers $[x^1_a|x^2_a|\cdots|x^n_a]$ from $QA$, respectively.
\begin{equation}
    (I + \Delta W)[X_q | \hat\beta X_a][X_q | \hat\beta X_a]^T 
    = [X_a | \hat\beta X_a][X_q | \hat\beta X_a]^T.
    \label{s3-eq1}
\end{equation}
We then expand the Eq. (\ref{s3-eq1}) to the following equation,
\begin{equation}
    \Delta W (X_q X_q^T + \hat\beta^2 X_a {X_a}^T) = X_a X_q^T - X_q X_q^T.
    \label{s3-eq2}
\end{equation}
Building upon Eq. (\ref{s3-eq2}), the $\Delta W$ can be directly solved by taking the inversion, thereby completing the proof.
%\begin{equation}
%    \Delta W = (X_a X_q^T - X_q X_q^T)(X_a {X_a}^T + X_q X_q^T)^{-1}.
%    \label{s3-eq3}
%\end{equation}
\end{proof}
While the Theorem \ref{7yujmpol} provides a closed-form solution for our edit operator ($W_{QA}=I+\Delta W$), the solution involves multiplying matrices $X_a {X_a}^T$ and $X_q X^T_q$. This can result in high computational overhead when working with a large number of candidate questions and answers. To mitigate this issue, we adopt the approach of using matrix decomposition to simplify the computation \cite{meng2022mass}.

\begin{remark}
    Denoting $n$ as the size of $QA$, $A = \hat{\beta}^2 n \sum_{i=1}^n \frac{1}{n} x_a^i (x_a^i)^T$ and $\lambda = \hat{\beta}^2 n > 0$ as a hyper-parameter, the matrix $\hat{\beta}^2 X_a (X_a)^T$ can be computed using $A$. Similarly, the matrix $X_q X_q^T$ can be computed using $Q = n \sum_{i=1}^n \frac{1}{n} x_q^i (x_q^i)^T$. Consequently, the following $\Delta \hat{W}$ is an efficient approximation to $\Delta W$, where $\lambda = \hat\beta^2 n$ is used to adjust the weight of maintaining the answer embedding $x_a$ constant (cf. Appendix \ref{hyper} for the analysis on $\lambda$).
    \begin{equation}
    \Delta \hat{W} = (X_a X^T_q - Q) (A + Q)^{-1}.
    \label{eq4}
\end{equation}
\end{remark}

% \todo{\textbf{CHEN}: Does this theorem 3.3 require a proof? If yes, add it. If not,  change the \{theorem\} to \{remark\}.}

\begin{remark}
    The derived edit operator $W_{QA}=I+(X_a X^T_q - Q) (A + Q)^{-1}$ retrofits $X_q$ without affecting $X_a$, ensuring that the DREditor can be used to DR models of various architectures.
\end{remark}

% Again, we use the $\lambda = \hat\beta^2 n $ to adjust the weight of holding the embedding of answers $x_a$ constant according to the dataset size $n$. 
% Further analysis on $\lambda$ is presented in Appendix \ref{hyper}. In this case, our edit operator is obtained by $W_{QA}=I+\Delta W$, which modifies the embedding of the question while ensuring time efficiency.
% \todo{because Eq (\ref{finaleq}) is weighted least squares form of Eq (\ref{eq3}), we can adjust the $\lambda$ to make the results of them more closer}
% \todo{$\lambda$ can be used to adjust the weight of the weighted least squares form in Theorem 3.1. This can further adjust the calibration to focus on keeping all the embedding of questions and answers unchanged or change most of them.}
To sum up, the edit operator, which is derived in closed form, can be expressed as $W_{QA}=I+(X_a X^T_q - Q) (A + Q)^{-1}$, and does not require an iterative optimization process. The computational complexity of $\Delta \hat{W}$ is $O(nd^2+d^3)$, where $d$ represents the dimension of an embedding vector. In practice, this complexity can be reduced by using advanced matrix inversion algorithms implemented in libraries such as Numpy and Pytorch. We highlight that the derived $W_{QA}$ does not affect the $x_a$ during the calibration (cf. Section \ref{intrin eva} for experiments). 
This enables our method to be applied to various DR models with different architectures, including single-tower and two-tower architectures.

% Due to the matrix multiplication and inversion, 
% Refer to Appendix \ref{yh23w98y239893} for more experimental discussion.

% and \todo{make the result of Eq (\ref{finaleq}) close to Eq (\ref{eq3})}, we can utilize the method suggested by a previous study \cite{meng2022mass}. This involves using $\mathbb{E}_{x_a\in QA} [x_a x_a^T]$ to estimate $X_a {X_a}^T$, which samples $x_a$ from $QA$.
% Due to the potentially unlimited size of $\{x^i_a\}\in QA$, obtaining $X_a$ from the overall $QA$ may not be feasible. 
% To address this, a previous study \cite{meng2022mass} suggests using $\lambda \mathbb{E}_{x_a\in QA} [x_a x_a^T]$ to estimate the $X_a {X_a}^T$, which involves sampling $x_a$ from $QA$.
% becoming Eq (\ref{eq4}). 
%Finally, the Eq (\ref{s3-eq3}) derives a closed-form solution as presented below.
%\begin{equation}
%    \Delta W = (X_q - X_a) X^T_q ( A +X_q X^T_q )^{-1}.
%    \label{eq4}
%\end{equation}
%Here, $X_q$ and $X_a$ are the matrix form of the embedding of questions from $QA$, respectively. $A = \lambda\mathbb{E}_{x_a \in QA}[x_a x^T_a]$, where $\lambda > 0$ is a hyper-parameter. The $\lambda$ is used to adjust the weight of holding the embedding of answers $x_a$ constant. Further analysis of this hyper-parameter is presented in Appendix \ref{hyper}. In this case, our edit operator is obtained by $W_{QA}=I+\Delta W$, which modifies the embedding of the question while ensuring efficiency in terms of time. 
% \todo{Need complexity analysis?}

\subsection{Efficient Embedding Calibration}
\label{thisiswhatisit}
The post-processing step of DREditor's matching rule editing involves a computation-efficient linear transformation, which is powered by the derived edit operator $W_{QA}$. This transformation adds an inductive bias  $\Delta Wx$  to the output embeddings $x$ of the off-the-shelf DR model, acting as a calibration term. By doing so, both question and answer embeddings are translated into a domain-specialized space, allowing the DR model to generalize to the target domain. DREditor offers a lightweight and portable approach to customize and store multiple edit operators for each enterprise efficiently, making it suitable for various domains. Algorithm \ref{alg:algorithm1} shows the pseudo-code of the algorithm, and it is noted that editing $x_a$ has little effect due to the invariance property. However, editing both $x_q$ and $x_a$ ensures alignment with the optimization objective, which guarantees theoretical effectiveness. For more information on the effectiveness of calibrating on $x_a$, please refer to Appendix \ref{onesizebaby}.

%The matching rule editing of DREditor is realized in a post-processing step. It only requires performing a computation-effective linear transformation, powered by the derived edit operator $W_{QA}$, on the output embeddings of the off-the-shelf DR model. Superficially, when applying $W_{QA}$ to the embedding $x$, it essentially adds an inductive bias $\Delta Wx$ to the embedding of $x$, acting as a calibration term. By this means, the embeddings of both questions and answers are efficiently translated into the domain-specialized space and hence generalizing the DR model to the target domain. In response to the requirements of the ES scenario, DREditor provides a portable, lightweight approach to power multiple enterprises across various domains by customizing and storing multiple edit operators for each enterprise efficiently. Formally, the pseudo-code of our algorithm is shown in Algorithm \ref{alg:algorithm1}. Notably, whether to edit $x_a$ or not will not affect the effect much due to the invariance property. However, we prefer editing both $x_q$ and $x_a$ to ensure alignment with our optimization objective and guarantee its effectiveness theoretically. Refer to Appendix \ref{onesizebaby} for our experimental comparison of our effectiveness with or without calibrating on $x_a$.

\begin{algorithm}[t]
    \caption{{Pseudo-code of DREditor}}
    \label{alg:algorithm1}
\begin{algorithmic}
\State\textbf{Input}: $\lambda > 0$, a domain-specific dataset $QA$, test dataset $QA^t$, a DR model
\State\textbf{Output}: Calibrated embedding
\end{algorithmic}
\begin{algorithmic}[1] %[1] enables line numbers
        \State \%\%\% \textit{Initial embeddings computation.}
        \State Obtain embeddings $(x^i_q, x^i_a) \in QA, \; (i =1,\cdots,n)$ from DR.
        % \State Obtain questions embedding matrix $X_q = [x^1_q|x^2_q|\cdots|x^n_q]$ and answers embedding matrix $X_a = [x^1_a|x^2_a|\cdots|x^n_a]$.
        % \State Estimate $X_aX_a^T$ by $A=\lambda\mathbb{E}_{x_a \in QA}[x_a x^T_a]$.
        \State \%\%\% \textit{Edit Operator Construction.}
        \State Compute $A = \hat{\beta}^2 n \sum_{i=1}^n \frac{1}{n} x_a^i (x_a^i)^T$, 
$Q = n \sum_{i=1}^n \frac{1}{n} x_q^i (x_q^i)^T$
        \State Compute $\Delta \hat{W} = (X_a X^T_q - Q) (A + Q)^{-1}$.
        \State Obtain the edit operator $W_{QA}=I+\Delta \hat{W}$.
        \State \%\%\% \textit{Efficient Embedding Calibration.}
        \State Obtain embeddings $(x^i_q, x^i_a) \in QA^t, \; (i =1,\cdots,p)$ from DR.
        \State Calibrate test embeddings by linear mapping $x'^{i}_q = x^i_q W_{QA}$ and $x'^{i}_a = x^i_a W_{QA} , \; (i =1,\cdots,p)$.
\end{algorithmic}
\end{algorithm}
% \setlength{\textfloatsep}{0.35cm}
% \setlength{\floatsep}{0.35cm}

% \begin{algorithm}[t]
%     \caption{{Pseudo-code of DREditor}}
%     \label{alg:algorithm}
% \begin{algorithmic}
% \State\textbf{Input}: $\lambda > 0$, a dataset $QA$, a DR model
% \State\textbf{Output}: Updated embedding
% \end{algorithmic}
% \begin{algorithmic}[1] %[1] enables line numbers
%         \State \%\%\% \textit{Editor operator construction.}
%         \State Construct $\overline{QA}$ by matching rule examination
%         \State Construct the edit operator $\Delta W$ by Eq (\ref{eq4}).
%         \State $W_{QA}=I+\Delta W$.
%         \State \%\%\% \textit{Efficient rule editing}
%         \State Editing adapter's parameter using $W_{QA}$.
%         \State \%\%\% \textit{DREditor in action}
%         \State Obtain embeddings of question $x_q$ and answer $x_a$ from DR. %each
%         \State Calibrating embeddings by $x_q' = x_q W_{QA}$ and $x_a' = x_a W_{QA}$.
% \end{algorithmic}
% \end{algorithm}

\textbf{Time Analysis.}
\label{timesec}
DREditor is a more efficient alternative to adapter-based finetuning because it directly calibrates the output embeddings through the edit operator instead of using iterative gradient-based backpropagation to modify the matching rule of DR.
The time consumption of our method can be divided into two parts. Firstly, solving the edit operator takes up most of the time in our method because it involves forward computation and matrix inverse operations. However, we experimentally find that our method can complete operations at the second and minute levels. Additionally, this time consumption is much lower than the time required for model fine-tuning.
Secondly, our rule editing step is applied as a post-processing step with two efficient matrix multiplications. This process only involves forward computation, which is equivalent to the testing stage of fine-tuned DR. However, to further reduce editing time, engineering-related tricks such as the cache mechanism can be used to cache the forward computation or DR embeddings. For our experiments, we chose not to use these tricks for fair comparison.

\textbf{Applied to ZeroDR Scenarios}. 
In this section, we demonstrate how DREditor can be used to edit the ZeroDR model for zero-shot domain-specific DR. This expands the application of DREditor to situations where enterprises on the ES platform lack domain-aware documents and struggle to build domain-specific models.
In particular, ZeroDR uses a source domain dataset to train and generalize a DR model to the target domain. Background knowledge-based methods, such as those mentioned in \cite{zhu2022enhancing, seonwoo2022virtual, rocktaschel2015injecting}, utilize datasets that provide an understanding of the specific concepts and explanations associated with the words or facts encountered in the target-domain text as the source-domain dataset. Similarly, we use background knowledge extracted from Wikipedia and ChatGPT as the source-domain dataset.
Since the background knowledge base can take different forms, structured or unstructured, we translate it into a unified representation of question-answer pairs $QA$, which is aligned with the data form of the text retrieval task. In our experiments, we consider both structured knowledge graphs extracted from WikiData and unstructured text facts from ChatGPT. Please refer to Algorithm \ref{alg:algorithm} and Section \ref{lalalalal} for more details.

\begin{algorithm}[t]
    \caption{{Pseudo-code of DREditor in ZeroDR Scenarios}}
    \label{alg:algorithm}
\begin{algorithmic}
\State\textbf{Input}: $\lambda > 0$, a test dataset $QA^t$, a DR model % , some domain key words $\{w_i\}_{i =1}^{k}$, \textbf{}
\State\textbf{Output}: Calibrated embedding
\end{algorithmic}
\begin{algorithmic}[1] %[1] enables line numbers
        \State \%\%\% \textit{Collect source-domain datasets.}
        \State Obtain source-domain pairs $QA$ from WikiData or ChatGPT
        \State \%\%\% \textit{Initial embeddings computation.}
        \State Obtain embeddings $(x^i_q, x^i_a) \in QA ,\; (i =1,\cdots,m)$ from DR.
        %\State Obtain questions embedding matrix $X_q = [x^1_q|x^2_q|\cdots|x^m_q]$ and  answers embedding matrix $X_a = [x^1_a|x^2_a|\cdots|x^m_a]$.
        %\State Estimate $X_aX_a^T$ by $A=\lambda\mathbb{E}_{x_a \in QA}[x_a x^T_a]$.
        \State \%\%\% \textit{Edit Operator Construction.}
        \State Compute $A = \hat{\beta}^2 n \sum_{i=1}^n \frac{1}{n} x_a^i (x_a^i)^T$, 
$Q = n \sum_{i=1}^n \frac{1}{n} x_q^i (x_q^i)^T$
        \State Obtain the edit operator $W_{QA}=I+(X_a X^T_q - Q) (A + Q)^{-1}$.
        \State \%\%\% \textit{Efficient Embedding Calibration.}
        \State Obtain embeddings $(x^i_q, x^i_a) \in QA^t, \; (i =1,\cdots,p)$ from DR.
        \State Calibrate test embeddings by linear mapping $x'^{i}_q = x^i_q W_{QA}$ and $x'^{i}_a = x^i_a W_{QA} , \; (i =1,\cdots,p)$.
\end{algorithmic}
\end{algorithm}

\section{Experiments}
We evaluate the efficiency and characteristics of DREditor across different domains.
To begin with, we conduct extrinsic experiments in Section \ref{asdrfasefcvasef} to demonstrate the calibration efficiency of DREditor and the retrieval performance of the domain-specific DR models that have been edited. These experiments showcase our superiority in various data sources, DR models, and computation resources such as GPUs and CPUs.
In addition, we perform intrinsic experiments to assess the properties of the calibrated embedding. In Section \ref{intrin eva}, we present the visualization of the embedding projections, which helps reveal a deeper understanding of our method. 

\subsection{Setup}
\label{lalalalal}
\textbf{Datasets from Different Domains \& Sources.}\footnote{See Appendix \ref{dataset} for more details. \label{apdC}} We conduct experiments on various domain-specific datasets to simulate the ES scenario, which requires powering multiple domains. To do this, we use BEIR \cite{thakur2021beir}, a benchmark for domain-specific DR. However, some datasets in BEIR are not publicly available or lack training data, so we exclude them from our experiments. Instead, we use FiQA-2018  \cite{10.1145/3184558.3192301} for financial domain evaluation, NFCorpus \cite{boteva2016full} for bio-medical information retrieval, and SciFact \cite{wadden-etal-2020-fact} for science-related fact-checking. To simulate scenarios where target-domain training data is lacking, we collect source-domain datasets for the aforementioned target domains from WikiData and ChatGPT. Refer to Appendix \ref{456789} for more details.

%To simulate the ES scenario that requires to power various domains, we experiment on various domain-specific datasets. To achieve this, we resort to BEIR \cite{thakur2021beir} in the experiments, which is a benchmark for domain-specific DR. Note that some datasets on BEIR are not publicly available or do not have {training data}, we omit those datasets on BEIR in our experiments. Consequently, we utilize FiQA-2018 \cite{10.1145/3184558.3192301} to conduct the financial domain evaluation, which is an opinion-based question-answering task. NFCorpus \cite{boteva2016full}, harvested from NutritionFacts, is used for the bio-medical information retrieval task. SciFact \cite{wadden-etal-2020-fact} is used for science-related fact checking, which aims to verify scientific claims using evidence from the research literature. To simulate the scenario that may lack of target-domain training data, we collect source-domain datasets for the above target-domain datasets from the WikiData and ChatGPT. 
\begin{itemize}[leftmargin=*]
    \item \underline{WikiData-based}. Given a sentence, we extract each entity $h$ and obtain the triplet $(h, r, t)$ and description information of $h$ and $t$ from WikiData. We then mask out head entity $h$ (or tail entity $t$) in a triple with the word "what" and concatenate the triple element together as a sentence $[S]$, which takes the form of "what $r$ $t$" or "$h$ $r$ what". Next, ChatGPT is utilized to paraphrase the sentence as a question. The corresponding answer is the description information of the masked entity.
    \item \underline{ChatGPT-based}. Inspired by \cite{petroni2019language, wang2020language}, we construct an open-domain background knowledge for each task using ChatGPT. In particular, given a sample text $[T]$ from test data (e.g., a query or a corpus), we require ChatGPT to offer factual information and explanations based on a template "\textit{Give me the factual information and explanations about the following statement or question: [T]}". The response, denoted as $[a]$, is further used to generate a question $q$ based on another template: "\textit{Read the following factual information and generate a question: $[a]$}". In this way, a question-answer pair that implies ChatGPT knowledge is obtained.
\end{itemize}

\noindent\textbf{DR and ZeroDR models}\footref{apdC}. Regarding the DR models, we consider the representative DR baselines on BEIR, including SBERT \cite{reimers-2019-sentence-bert}, DPR \cite{karpukhin-etal-2020-dense}, and ANCE \cite{xiongapproximate}. Each baseline is initialized by the public available checkpoint that is pre-trained based on an open-domain corpus, such as MS MARCO extracted from the Web and Natural Questions\cite{kwiatkowski-etal-2019-natural} mined from Google Search. Regarding the ZeroDR models, we involve the recent SOTA ZeroDR model, called COCO-DR \cite{yu-etal-2022-coco}, which pretrains the Condenser \cite{gao-callan-2021-condenser} on the corpus of BEIR and MS MARCO to enhance the domain generalization. 

\noindent\textbf{Baselines}\footref{apdC}. Our method is closely related to the recent studies on \underline{embedding calibration}. Since there is no existing calibration method that works on dense retrieval or text-matching tasks, we cannot make a comparison. Instead, we compare DREditor with the conventional \underline{adapter finetuning}\footnote{For a more comprehensive understanding, we also involve LoRA \cite{hu2021lora} as parameter-efficient finetuning baselines in Appendix \ref{yhnoilk23qweasdoi}.} to demonstrate our time efficiency. We consider two task scenarios: one where there is domain-specific data available, and another where there is no domain-specific data available (zero-shot). In the domain-specific DR experiments (cf. section \ref{yuhjn3eudj}), we fine-tune each DR baseline using the domain-specific training data. To ensure a fair comparison, we add one adapter layer to each DR baseline and only fine-tune this layer. In the domain-specific ZeroDR experiment (cf. section \ref{7uj2w3o8e34rweoh}), we follow the idea of incorporating background knowledge to enhance domain generalization of DR \cite{zhu2022enhancing, seonwoo2022virtual, rocktaschel2015injecting}. We use the background knowledge data to fine-tune the adapter of each DR baseline. Previous studies have shown that this approach can be effective in improving domain generalization in DR.

\noindent\textbf{Computation Resources}. We push the limits of DREditor and evaluate its efficiency when computational resources are limited. We conduct experiments to demonstrate how well DREditor performs on different computation devices, specifically GPUs and CPUs. For all of our experiments, we use an Nvidia A6000 48GB GPU and an Intel (R) Xeon (R) Gold 6348 CPU @ 2.60GHz. The only difference between the experiments is the type of computing device used; all other hardware and software settings are kept constant.

\noindent\textbf{Evaluation Metrics.} We value and assess the time efficiency of the DR customization to the given domain and its retrieval performance on that domain.
To assess time efficiency, we consider the computing time among different computation devices (i.e., GPU and CPU) when performing DR/ZeroDR model finetuning or embedding calibration.
Our method's time calculation includes the entire computing process, including edit operator construction, as described in Algorithm \ref{alg:algorithm1} and Algorithm \ref{alg:algorithm}. In contrast, the time calculation for adapter finetuning includes the time spent on the training and validation process, including gradient descent and backpropagation. As for the retrieval performance evaluation, we follow previous studies \cite{yu-etal-2022-coco, santhanam2022colbertv2} and report the nDCG@10 (cf. Appendix \ref{yhnoilk23qweasdoi} for more results on other performance metrics).

\subsection{Calibration Efficiency on DR and ZeroDR} 
\label{asdrfasefcvasef}
% Addressing the rapid deployment issue of dense retrieval in different domains of CCS, 
This section aims to show our superiority in different domain-specific datasets (i.e., finance, science, and bio-medicine), data sources (i.e., DR and ZeroDR), DR models, and computing devices (GPU and CPU) in terms of time efficiency and retrieval performance. 

% To conduct a comprehensive empirical evaluation, we experiment on three domain-specific tasks, inclusive of financial question answering, bio-medical information retrieval, and science-related fact-checking. Meanwhile, for each task, we examine the domain generalizability under two settings, i.e., general DR and zero-shot DR (ZeroDR). In this paper, general DR means to utilize specific-domain training data to establish a domain-specific DR model (cf. section \ref{yuhjn3eudj}), while ZeroDR does not access to the training data from the target domain \cite{thakur2021beir, neelakantan2022text, izacard2022unsupervised}. Instead, it trains on source-domain data to create a DR model that can generalize well to the target domain (cf. section \ref{7uj2w3o8e34rweoh}). 

\subsubsection{Domain-specific DR evaluation}
\label{yuhjn3eudj}
Table \ref{tab:dr} shows the results of domain-specific DR evaluation in comparison to adapter finetuning. These results verify the potential of our DREditor to efficiently and inexpensively edit the DR without significantly impacting retrieval performance. We also included LoRA \cite{hu2021lora} as a parameter-efficient finetuning baseline, but since it prioritizes effective training by reducing parameters and improving device usage rather than time efficiency, we provide the comparison results in Appendix \ref{yhnoilk23qweasdoi}. Basically, our results are consistent with previous studies \cite{guo2020parameter, houlsby2019parameter, fu2023effectiveness} that show LoRA requires more training time than adapter finetuning, let alone compared to DREditor.

% It achieves a promising retrieval performance that is comparable to the adapter finetuning, while also providing a substantial improvement in time efficiency of nearly 300 times, averaged on different devices and tasks.

\begin{table}[t]
\centering
\caption{Domain-specific DR evaluation. DREditor efficiently and inexpensively adapts the DR to a given domain without significantly impacting retrieval performance.}
\label{tab:dr}
% \vspace{-1mm}
\resizebox{0.49\textwidth}{!}{%
\begin{tabular}{l|r|c|r|r}
\toprule
\multicolumn{1}{l|}{\textbf{Backbone}} & \multicolumn{1}{c|}{\textbf{Data + Method}} & \multicolumn{1}{c|}{\textbf{Perf. (\%)}} & \multicolumn{1}{c|}{\textbf{GPU Time}} & \multicolumn{1}{c}{\textbf{CPU Time}} \\ \hline
\multirow{9}{*}{SBERT} & \multicolumn{1}{l|}{SciFact} & 55.48 & -- & -- \\
 & + Fine-tune & 68.05$_{+12.57}$ & 2min58s & 28min15s \\
 & + DREditor & \textbf{70.18$_{+14.70}$} & \textbf{7.09s} & \textbf{35.32s} \\ \cline{2-5}
 & \multicolumn{1}{l|}{FiQA} & 23.41 & -- & -- \\
 & + Fine-tune & \textbf{29.45$_{+6.04}$} & 33min15s & 3h26min40s \\
 & + DREditor & 26.18$_{+2.77}$ & \textbf{9.21s} & \textbf{6min23s} \\ \cline{2-5}
 & \multicolumn{1}{l|}{NFCorpus} & 27.36 & -- & -- \\
 & + Fine-tune & \textbf{33.12}$_{+5.76}$ & 2h26min22s & 1d22h15min20s \\
 & + DREditor & 30.30$_{+2.94}$ & \textbf{1min16s} & \textbf{1h9min39s} \\ \midrule
\multirow{9}{*}{DPR} & \multicolumn{1}{l|}{SciFact} & 19.70 & -- & -- \\
 & + Fine-tune & \textbf{58.88}$_{+39.18}$ & 10min10s & 1h36min45s \\
 & + DREditor & 54.35$_{+34.65}$ & \textbf{12.19s} & \textbf{36.67s }\\ \cline{2-5}
 & \multicolumn{1}{l|}{FiQA} & 8.62 & -- & -- \\
 & + Fine-tune & \textbf{12.92}$_{+4.3}$ & 3h53min50s & 1d7h42min20s \\
 & + DREditor & 10.10$_{+1.48}$ & \textbf{13.28s} & \textbf{6min24s} \\ \cline{2-5}
 & \multicolumn{1}{l|}{NFCorpus} & 17.06 & -- & -- \\
 & + Fine-tune & 16.77$_{-0.29}$ & 1d5h57min20s & 7d23h9min16s \\
 & + DREditor & \textbf{18.92}$_{+1.86}$ & \textbf{43.16s} & \textbf{17min57s} \\ \midrule
\multirow{9}{*}{ANCE} & \multicolumn{1}{l|}{SciFact} & 51.72 & -- & -- \\
 & + Fine-tune & \textbf{51.93}$_{+0.21}$ & 5min58s & 24min40s \\
 & + DREditor & 51.88$_{+0.16}$ & \textbf{12.53s} & \textbf{1min20s} \\ \cline{2-5}
 & \multicolumn{1}{l|}{FiQA} & 27.09 & -- & -- \\
 & + Fine-tune & 28.22$_{+1.13}$ & 48min37s & 6h23min30s \\
 & + DREditor & \textbf{28.41}$_{+1.32}$ & \textbf{15.66s} & \textbf{13min39s} \\ \cline{2-5}
 & \multicolumn{1}{l|}{NFCorpus} & 22.01 & -- & -- \\
 & + Fine-tune & \textbf{26.37}$_{+4.36}$ & 2h20min21s & 1d22h15min10s \\
 & + DREditor & 23.44$_{+1.43}$ & \textbf{2min13s} & \textbf{2h36min14s}  \\ \bottomrule
\end{tabular}%
}
\vspace{-3mm}
\end{table}

\begin{table*}[!htb]
\centering
\caption{Domain-specific ZeroDR evaluation. Here, 'FT' means finetuning, 'KG' and 'ChatGPT' represent the structured knowledge graph from WikiData and unstructured text facts from ChatGPT, respectively. DREditor is vastly superior to the adapter finetuning in both zero-shot retrieval performance and time efficiency.}
\label{tab:zerord}
% \vspace{-1mm}
\resizebox{0.999\textwidth}{!}{%
\begin{tabular}{l|l|c|c|c|c|c|c|c|c|c}
\toprule
\multicolumn{1}{l|}{\multirow{3}{*}{\textbf{Backbone}}} & \multicolumn{1}{l|}{\multirow{3}{*}{\textbf{Methods}}} & \multicolumn{3}{c|}{\textbf{Perf. (\%)}} & \multicolumn{6}{c}{\textbf{Training / Calibration Time}} \\ \cline{3-11} 
\multicolumn{1}{l|}{} & \multicolumn{1}{l|}{} & \multicolumn{1}{c|}{\multirow{2}{*}{\textbf{SciFact}}} & \multicolumn{1}{c|}{\multirow{2}{*}{\textbf{FiQA}}} & \multicolumn{1}{c|}{\multirow{2}{*}{\textbf{NFCorpus}}} & \multicolumn{2}{c|}{\textbf{SciFact}} & \multicolumn{2}{c|}{\textbf{FiAQ}} & \multicolumn{2}{c}{\textbf{NFCorpus}} \\ \cline{6-11} 
\multicolumn{1}{l|}{} & \multicolumn{1}{l|}{} & \multicolumn{1}{c|}{} & \multicolumn{1}{c|}{} & \multicolumn{1}{c|}{} & \multicolumn{1}{l|}{\textbf{CPU}} & \multicolumn{1}{l|}{\textbf{GPU}} & \multicolumn{1}{l|}{\textbf{CPU}} & \multicolumn{1}{l|}{\textbf{GPU}} & \multicolumn{1}{l|}{\textbf{CPU}} & \textbf{GPU} \\ \midrule
\multirow{4}{*}{SBERT} & Ours$_{KG}$ & \textbf{55.58} & \textbf{23.39} & \textbf{27.47} & \textbf{2min5s} & \textbf{1min19s} & \textbf{2min29s} & \textbf{1min35s} & \textbf{1min24s} & \textbf{47.45s} \\
 & FT$_{KG}$ & 45.02 & 20.26 & 22.12 & 43min20s & 5min37s & 7h6min20s & \multicolumn{1}{c}{25min11s} & \multicolumn{1}{c}{50min10s} & \multicolumn{1}{c}{7min55s} \\ \cline{2-11} 
 & Ours$_{ChatGPT}$ & 56.12 & \textbf{23.74} & \textbf{27.39} & \textbf{6.82s} & \textbf{3.73s} & \textbf{13.75s} & \textbf{7.86s} & \textbf{7.97s} & \textbf{3.95s} \\
 & FT$_{ChatGPT}$ & \textbf{56.14} & 22.47 & 26.56 & 1min50s & 10.65s & 2h50min20s & 16min22s & 24min30s & 2min11s \\ \midrule
\multirow{4}{*}{DPR} & Ours$_{KG}$ & \textbf{24.27} & \textbf{9.65} & \textbf{18.10} & \textbf{3min54s} & \textbf{2min22s} & \textbf{4min38s} & \textbf{2min48s} & \textbf{2min30s} & \textbf{1min26s} \\
 & FT$_{KG}$ & 13.30& 5.27& 13.09& 27h35min38s & 4h31min19s & 28h53min20s & 4h33min11s & 20h8min42s & 7h27min39s \\ \cline{2-11} 
 & Ours$_{ChatGPT}$ & \textbf{23.25} & \textbf{9.45} & \textbf{17.26} & \textbf{12.39S} & \textbf{6.78S} & \textbf{31.41S} & \textbf{14.82S} & \textbf{13.58S} & \textbf{7.29S} \\ 
 & FT$_{ChatGPT}$ & 12.43& 7.34& 14.66 & 37min36s & 11min39s & 1h16min24s  & 8min7s & 35min45s & 4min23s \\  \midrule
\multirow{4}{*}{ANCE} & Ours$_{KG}$ & \textbf{51.85} & \textbf{27.53} & \textbf{22.14} & \textbf{3min59s} & \textbf{2min27s} & \textbf{4min28s} & \textbf{2min49s} & \textbf{2min44s} & \textbf{1min27s} \\
 & FT$_{KG}$ & 26.90 & 13.98 & 13.34 & 1h23min & 9min1s & 7h52min50s & 44min29s & 2h27min50s & 13min9s \\ \cline{2-11} 
 & Ours$_{ChatGPT}$ & \textbf{51.83} & \textbf{27.77} & \textbf{22.05} & \textbf{12.37s} & \textbf{7.01s} & \textbf{34.81s} & \textbf{14.79s} & \textbf{13.77s} & \textbf{7.46s} \\
 & FT$_{ChatGPT}$ & 39.34 & 18.06 & 19.18 & 3min30s & 20.23s & 6h36min12s & 31min22s & 47min50s & 4min2s \\ \midrule
 
\multicolumn{2}{c|}{COCO-DR (\textit{SOTA ZeroDR})} & 66.36 & 25.96 & 17.63 & \multicolumn{6}{c}{*1.5 Days (8 A100 80GB GPUs \& FP16 mixed-precision) \cite{yu-etal-2022-coco}} \\ \bottomrule
\end{tabular}%
}\vspace{-1mm}
\end{table*}

\textbf{DREditor offers substantial enhancements in time efficiency, achieving speeds up to 300 times faster while maintaining retrieval performance without significant compromise}.
Both adapter finetuning and DREditor significantly improve retrieval performance compared to the backbone DR models across all three domain-specific tasks.
Specifically, DREditor achieves an average improvement of 7.43\% on all backbones and datasets, while adapter finetuning achieves an average improvement of 9.19\%.
This suggests that both adapter finetuning and DREditor can modify the matching rules of the backbone DR model and achieve retrieval performance gains on domain-specific data.
Meanwhile, DREditor has a significant impact on time efficiency, especially when compared to the finetune-based approach. The improvements are remarkable, with DREditor providing a computing efficiency enhancement of around 100 times on the GPU and 400 times on the CPU when compared to adapter finetuning. On average, this translates to a nearly 300 times improvement in time efficiency across all tasks and computation devices.
Additionally, DREditor on the CPU shows a large improvement in time efficiency when compared to GPU-based model finetuning. %For example, DREditor only took 17 minutes of CPU time to edit the matching rules of DR on NFCorpus data, while finetuning on the same data required more than 7 days of CPU time (i.e., 600 times improvement) or one hour of GPU time (i.e., 3.5 times improvement).

Our study demonstrates the significant potential of DREditor to edit DR in an inexpensive and time-efficient manner while maintaining retrieval performance. This is especially beneficial for individuals and small ES providers who create domain-specific DR models in low-resource settings, such as CPUs. By reducing the need for expensive computing devices when editing the matching rules of DR, DREditor offers a practical solution for those who require efficient DR editing capabilities but have limited access to high-priced computing resources. Overall, our findings suggest that DREditor is a valuable tool for improving the efficiency and accessibility of domain-specific DR models.

% Our findings highlight the tremendous potential of our DREditor method to inexpensively edit the DR while maintaining retrieval performance. This is especially important for individual users and small institutions that construct domain DR models in low resource scenarios, such as CPUs. DREditor can significantly reduce the need for expensive computing devices when editing the matching rules of DR, without significantly compromising retrieval performance. This provides a practical solution for those who have limited access to high-priced computing resources but still require efficient DR editing capabilities.

\subsubsection{Domain-specific ZeroDR evaluation}
\label{7uj2w3o8e34rweoh}
%We follow the idea of incorporating the knowledge base to enhance domain generalization of DR \cite{zhu2022enhancing, seonwoo2022virtual, rocktaschel2015injecting}. 

We utilize background knowledge data as the source domain data and evaluate the performance of the target-domain test data. We refer to the knowledge data from the structured WikiData knowledge as 'KG' and the unstructured text facts from ChatGPT as 'ChatGPT'. Each knowledge dataset is processed through DREditor to derive a specialized edit operator. Additionally, we use the dataset to fine-tune our backbone DR models for comparison purposes. We also include the recently open-sourced state-of-the-art ZeroDR model, called COCO-DR \cite{yu-etal-2022-coco}, in our experiment. It's important to note that our primary goal is to explore a more time-efficient method for any DR model rather than proposing a ZeroDR model with better retrieval performance.
%We utilize the collected background knowledge data as the source domain data and evaluate the performance of the target-domain test data. In this paper, we denote the knowledge data from the structured knowledge graph as 'KG' and unstructured text facts as 'ChatGPT'. Each knowledge dataset is fed into DREditor to derive the specialized edit operator. Additionally, we utilize the dataset to fine-tune our backbone models for comparison. Recently open-source state-of-the-art ZeroDR model, called COCO-DR \cite{yu-etal-2022-coco}, is also involved in our experiment. It is important to note that our goal is to explore a more time efficient method for any DR model, rather than proposing a ZeroDR model with better retrieval performance.  

\begin{figure*}[t]
    \centering
    \includegraphics[width=0.95\textwidth]{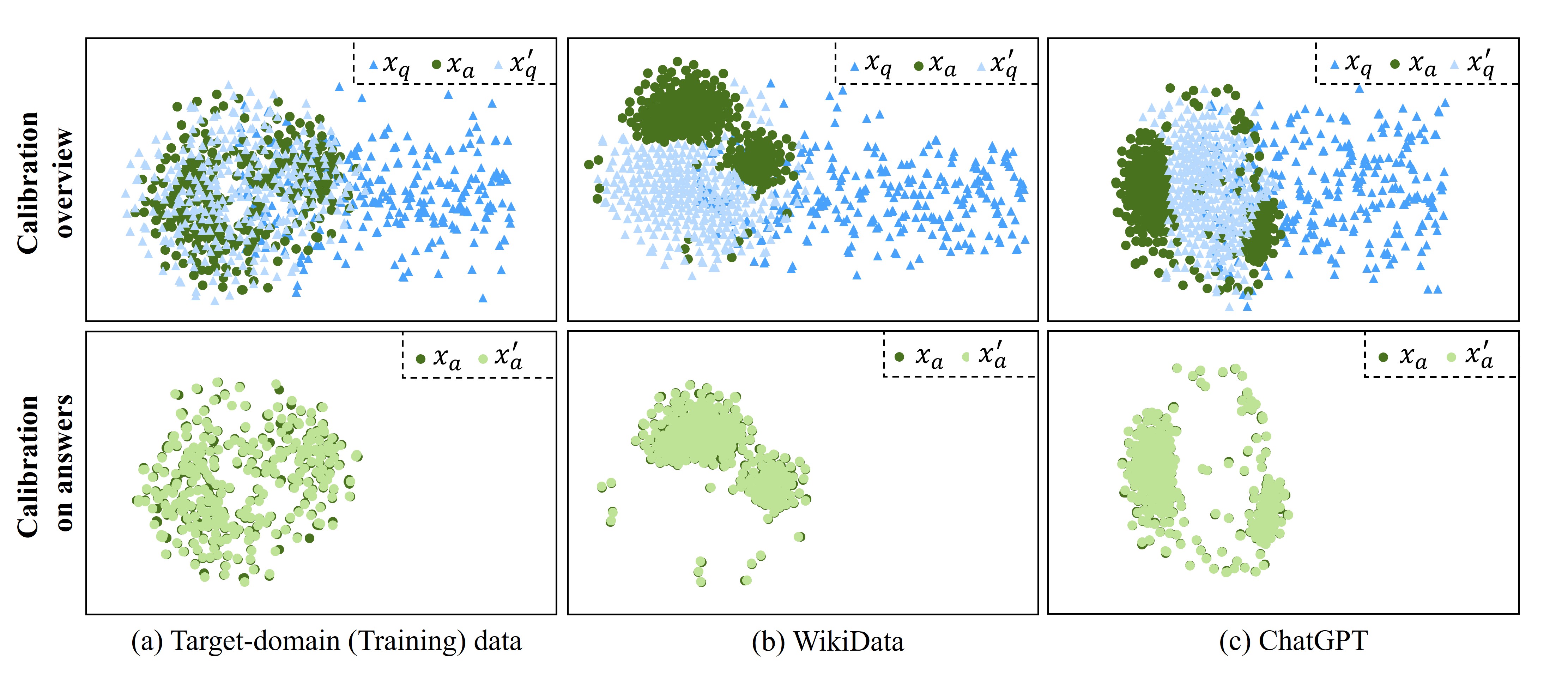}
    % \vspace{-1mm}
    \caption{Embedding illustration on SciFact with SBERT. ${x_q}'$ and ${x_a}'$ as calibrated question and answer embeddings. DREditor enhances the semantic association between the questions and answers while preserving the answer embeddings unchanged.
    }
    \label{fig:scifact-sbert}
    \vspace{-3mm}
\end{figure*}

\textbf{DREditor outperforms adapter finetuning significantly in both zero-shot retrieval performance and time efficiency.}.
As presented in Table \ref{tab:zerord}\footnote{In our experiment, the size of the knowledge data from ChatGPT is smaller than from WikiData. As a result, training or editing on the ChatGPT-based data saves more computational time.}, on average, our method outperforms the adapter finetuning by an impressive 7.19 across all tasks and backbone models. Furthermore, our method boasts an exceptional level of time efficiency, achieving an improvement of almost 100 times across all tasks and computation devices. This means that ES providers can enjoy faster and more accurate retrieval performance without the need for high-priced computing devices. These results demonstrate the effectiveness of our approach and its potential to efficiently edit the DR model and adapt it to the zero-shot scenario. In particular, our DREditor continues to deliver impressive results, achieving an average performance gain of 11.62 on SciFact, 5.69 on FiQA, and 4.24 on NFCorpus. These findings indicate that our approach is transferable across different models. One possible explanation for the inadequacy of finetuning is that the DR backbone may overfit the training data, which has a distribution that is not aligned with the test data in zero-shot scenarios. This distribution mismatch could lead to poor performance in zero-shot retrieval tasks \cite{yu-etal-2022-coco}.
When it comes to CPU computation, DREditor boasts an improvement in time efficiency of around 200 times and 70 times on GPU computation. 
More interestingly, while our approach only requires one A6000 GPU or a CPU and takes a few minutes to edit the DR and adapt it to zero-shot scenarios, COCO-DR demands 8 A100 GPUs and takes almost two days to train. However, such time costs do not necessarily bring constant performance gains on retrieval performance. Even though COCO-DR has already utilized BEIR's corpus (including our used datasets) in its pre-training stage, it still struggles to achieve better results on tasks other than SciFact.

\textbf{ChatGPT has a higher tendency to produce background knowledge of higher quality compared to WikiData in ZeroDR}. 
When we compared the performance of DREditor using different knowledge bases, we found that the two knowledge bases had little effect on retrieval performance, with an average difference of around 0.12. However, the choice of knowledge base had a significant impact on adapter finetuning. Adapter finetuning performed better on ChatGPT data than on KG, with an increase of 4.76. This suggests that the background knowledge from ChatGPT may be of higher quality than that from WikiData.

% Comparing DREditor with different knowledge bases, we discover that the two knowledge base data had little impact on the retrieval performance of DREditor, with an average difference of approximately 0.12. However, it has a significant impact on adapter finetuning. The adapter finetuning performs better on ChatGPT data than on KG (increased by 4.76). This suggests that background knowledge from ChatGPT may be higher-quality than that from WikiData.

\subsection{Visualization of the Embedding Projections}
\label{intrin eva}
% \todo{What is our expected conclusion? 1. invariance to a; 2) illustrate the matching rule is indeed changed. the success of our method is due to the this}
This section aims to perform intrinsic experiments that reveal and confirm the properties of DREditor. Specifically, we use t-SNE \cite{van2008visualizing} to visualize the embedding projections of questions and answers with and without DREditor's calibration. The calibration is done by edit operators using different data sources, and we denote ${x_q}'$ and ${x_a}'$ as the question and answer embedding after calibration, respectively. The results are shown in Figure \ref{fig:scifact-sbert}. For more visualizations, please refer to Appendix \ref{morereulsts}.

% illustrate how the embedding space of question and answer is corrected by our DREditor, using visualization of the embedding projections.
%Since DREditor directly calibrates the embeddings of the DR models, we start our analysis based on the calibrated embedding space. Specifically, we compare the differences of both question and answer embeddings with or without the calibration operation. Here, we denote ${x_q}'$ and ${x_a}'$ as the question and answer embedding after calibration, respectively. We drew the SBERT\footnote{We leave the results of other models in Appendix \ref{morereulsts}. \label{apdmorereulsts}} embeddings after t-SNE in Fig.\ref{fig:scifact-sbert}. 

% \textbf{Main results}. Fig.\ref{fig:scifact-sbert} demonstrates three important observations: 

\textbf{DREditor translates the output embeddings of a DR model to the specialized space with a decreased mean discrepancy between the embeddings of the questions and answers}. 
As per the top row of Figure \ref{fig:scifact-sbert}, the embedding distributions of questions and answers become closer after embedding calibration. This is achieved by only retrofitting the question embeddings while maintaining invariance towards the answer embeddings, as shown in the second row of Figure \ref{fig:scifact-sbert}. Considering the improved retrieval performance in Table \ref{tab:dr} and Table \ref{tab:zerord}, we can safely conclude that the reduced mean discrepancy is due to the improved semantic similarity between the questions and their corresponding answers.

% 1) After DREditor edited the DR model, the calibrated question embedding and answers embedding became closer together. 2) Such the calibration operation hardly affects the embeddings of answers. These two observations empirically reveal that our DREditor, as discussed in section \ref{7yuqw34ukryqeawib}, successfully edits the matching rule and enhances the association between the questions and answers while preserving the representation of corpus unchanged. 

\textbf{Enhancing the correlation between the data used for correction and the target-domain data can help reduce their distribution shift}.
The top row of Figure \ref{fig:scifact-sbert} shows that question embeddings calibrated by different data sources are distributed differently.
In-domain calibrated question embeddings are more uniformly distributed in the space of answer embeddings, whereas question embeddings calibrated by source-domain knowledge (as used in Section \ref{7uj2w3o8e34rweoh}) are denser and more clustered.
Our method migrates the data distribution from the source domain to the target domain as much as possible, which explains why it is effective in zero-shot scenarios. However, calibrating embeddings based on source-domain knowledge cannot achieve distribution alignment with the target domain. Using data with rich in-domain information to calibrate embeddings and edit the matching rules of the DR model is always a better choice. Our experiments verify that utilizing source-domain data to efficiently calibrate DR embeddings is also a promising solution for the zero-shot setting. In any case, we have successfully implemented and verified the efficiency and effectiveness of our approach through empirical experiments. We believe that our work provides new insights for subsequent domain-specific DR research.

\section{Conclusion}
We developed a general toolkit for performing targeted post-hoc modifications to dense retrieval models, with the goal of building a domain-specific dense retrieval model time-efficiently. Crucially, instead of specifying the desired behavior implicitly via long-time finetuning, our method allows users to directly edit the model’s matching rules. By doing so, our approach makes it easier for enterprise search providers to customize a domain-aware search engine to secure the scalability of various enterprises and meet the time-demanding requirements of these enterprises. Additionally, a key benefit of our method is that it fills the blank by incorporating the emerging research studies on embedding calibration into the retrieval task.
Finally, our method also is applicable to zero-shot domain-specific DR scenarios, where we calibrate the embedding using source-domain background knowledge. 
In a broader context, our approach has explored the potential of editing the matching rule of dense retrieval models with embedding calibration. We believe that this opens up new avenues for building domain-specific DR models efficiently and inexpensively.

%\clearpage

\bibliographystyle{ACM-Reference-Format}
\bibliography{ref}

\appendix

%Due to the potentially unlimited size of $\{x^i_a\}\in QA$, obtaining $X^{QA}_a$ from the overall $QA$ may not be feasible. 
%To address this, a previous study \cite{meng2022mass} suggests using $\lambda \mathbb{E}_{x_a\in QA} [x_a x_a^T]$ to estimate the $X^{QA}_a {X^{QA}_a}^T$, which involves sampling $x_a$ from $QA$.
%This results in Eq (\ref{ap-eq2}) becoming Eq (\ref{eq4}). The $\lambda$ is used to adjust the weight of holding the embedding of answers $x_a$ constant. Further analysis of hyper-parameters is presented in Appendix \ref{hyper}.
\section{The proof of theorem \ref{theory}}
\label{proof}

\begin{thmbis}{theory}
 The optimization problem (\ref{eq3}) could be rephrased as the following form that involves only $QA$, where $\hat\beta$ is a hyper-parameter associated with $\beta$. When $\beta \gg 1$, $\hat\beta \approx \beta$.
    \begin{equation}
        \begin{split}
    \Delta W = & \arg \min_{\overline{W}} \sum_{(x_q, x_a) \in QA} \| (I + \overline{W})x_q - x_a \|_2^2 \\
    & +  \hat\beta \sum_{x_a \in QA} \|  (I + \overline{W})x_a - x_a \|_2^2.
        \end{split}
    \label{proof1}
    \end{equation}
\end{thmbis}

\begin{proof}
First, we can rewrite the Eq (\ref{proof1}) to the following form:
    \begin{equation}
        \begin{split}
    \Delta W = & \arg \min_{\overline{W}} \sum_{(x_q, x_a) \in \overline{QA}} \| (I + \overline{W})x_q - x_a \|_2^2 \\
     & + \sum_{(x_q, x_a) \in QA/\overline{QA}} \| (I + \overline{W})x_q - x_a \|_2^2\\
    &+ \hat\beta\sum_{x_a \in \overline{QA}} \|  (I + \overline{W})x_a - x_a \|_2^2\\
    &+\hat\beta\sum_{x_a \in QA/\overline{QA}} \|  (I + \overline{W})x_a - x_a \|_2^2.
        \end{split}
    \label{proof2}
    \end{equation}

Here, we separate the $QA$ of the Eq (\ref{proof1}) into $\overline{QA}$ part and $QA/\overline{QA}$ part. The $QA/\overline{QA}$ part means the original correct question-answer pairs, which further implies 
\[when \quad (x_a,x_q) \in QA/\overline{QA},\quad \| x_q - x_a \|_2^2 \quad is \quad small. \]

Similarly, we can rewrite the Eq (\ref{eq3}) to the following form:
    \begin{equation}
        \begin{split}
    \Delta W = & \arg \min_{\overline{W}} \sum_{(x_q, x_a) \in \overline{QA}} \| (I + \overline{W})x_q - x_a \|_2^2 \\
    &+ \beta\sum_{x_a \in \overline{QA}} \|  (I + \overline{W})x_a - x_a \|_2^2\\
    &+\beta\sum_{x_a \in QA/\overline{QA}} \|  (I + \overline{W})x_a - x_a \|_2^2.
        \end{split}
    \label{proof3}
    \end{equation}

Therefore, the difference between Eq (\ref{eq3}) and Eq (\ref{proof1}) is the second term of Eq (\ref{proof2}):
 \[\sum_{(x_q, x_a) \in QA/\overline{QA}} \| (I + \overline{W})x_q - x_a \|_2^2.\]

Because the $QA/\overline{QA}$ is special, we can have this approximation:
    \begin{equation}
        \begin{split}
\sum_{(x_q, x_a) \in QA/\overline{QA}} \| (I + \overline{W})x_q - x_a \|_2^2 \approx  \\
\sum_{(x_q, x_a) \in QA/\overline{QA}} \| (I + \overline{W})x_a - x_a \|_2^2.        
\end{split}
    \label{proof4}
    \end{equation}

Then the Eq (\ref{proof1}) becomes to the following:
    \begin{equation}
        \begin{split}
    \Delta W = & \arg \min_{\overline{W}} \sum_{(x_q, x_a) \in \overline{QA}} \| (I + \overline{W})x_q - x_a \|_2^2 \\
    &+ \hat\beta \sum_{x_a \in \overline{QA}} \|  (I + \overline{W})x_a - x_a \|_2^2\\
    &+ (\hat\beta+1) \times \sum_{x_a \in QA/\overline{QA}} \|  (I + \overline{W})x_a - x_a \|_2^2.
        \end{split}
    \label{proof5}
    \end{equation}

 This is a reweighted form of Eq (\ref{proof3}). When $\beta \gg 1$, $\hat\beta \approx \beta \approx \hat\beta+1$, which makes Eq (\ref{proof3}) almost equal Eq (\ref{proof5}).
\end{proof}

\section{Appendix for Experimental Results}
\subsection{More Results of Calibration Efficiency on DR}
\label{yhnoilk23qweasdoi}
To give a more comprehensive performance evaluation, we report additional metrics (Top-10 Mean Average Precision and Top-10 Recall, i.e., MAP@10 and Recall@10) to evaluate our method. 
In addition, although LoRA \cite{hu2021lora} prioritizes effective training by reducing parameters to optimize and improve the device usage, failing to improve time efficiency \cite{guo2020parameter, houlsby2019parameter, fu2023effectiveness}, we also involve it into the comparison for better understanding. 
We conduct the experiments in the setting of domain-specific DR, and the results are shown in Table \ref{tab:MResultsExtrinsic}. 

% In addition, since LoRA is also a low-parameter modification method at present, although it has a different goal from ours of saving time, we also demonstrate it here.

% Table generated by Excel2LaTeX from sheet 'Sheet12'
\begin{table}[!htb]
  \centering
  \caption{More Results of Calibration Efficiency on DR}
    \resizebox{0.45\textwidth}{!}{
    \begin{tabular}{c|r|c|c|c|r|r}
    \toprule
    \textbf{Backbone} & \multicolumn{1}{c|}{\textbf{Data+Method}} & \textbf{NDCG@10} & \textbf{MAP@10} & \textbf{Recall@10} & \multicolumn{1}{c|}{\textbf{GPU Time}} & \multicolumn{1}{c}{\textbf{CPU Time}} \\
    \midrule
    \multirow{15}[6]{*}{SBERT} & \multicolumn{1}{l|}{SciFact} & 55.48 & 50.40 & 68.92 & --    & \multicolumn{1}{r}{--} \\
          & + Finetune & 68.05 & 63.68 & 79.97 & 2min58s & \multicolumn{1}{r}{28min15s} \\
          & + LORA & 66.59 & 61.53 & \textbf{81.04} & 3min2s & \multicolumn{1}{r}{--} \\
          & + DREditor & \textbf{72.07} & \textbf{70.00} & 77.26 & \textbf{7s} & \multicolumn{1}{r}{\textbf{35s}} \\
\cmidrule{2-7}          & \multicolumn{1}{l|}{FiQA} & 23.41 & 17.77 & 29.02 & --    & \multicolumn{1}{r}{--} \\
          & + Finetune & 29.10 & \textbf{22.31} & 36.26 & 33min15s & \multicolumn{1}{r}{3h26min40s} \\
          & + LORA & \textbf{29.48} & 22.27 & \textbf{37.23} & 28min41s & \multicolumn{1}{r}{--} \\
          & + DREditor & 25.22 & 19.18 & 31.06 & \textbf{9s} & \multicolumn{1}{r}{\textbf{383s}} \\
\cmidrule{2-7}          & \multicolumn{1}{l|}{NFCorpus} & 27.36 & 9.38 & 13.27 & --    & \multicolumn{1}{r}{--} \\
          & + Finetune & 32.60 & 10.94 & 15.90 & 2h26min22s & \multicolumn{1}{r}{1d22h15min20s} \\
          & + LORA & \textbf{37.16} & \textbf{15.00} & \textbf{19.17} & 2h21min10s & \multicolumn{1}{r}{--} \\
          & + DREditor & 30.37 & 10.62 & 14.23 & \textbf{1min16s} & \multicolumn{1}{r}{\textbf{1h9min39s}} \\
    \midrule
    \multirow{12}[6]{*}{DPR} & \multicolumn{1}{l|}{SciFact} & 19.70 & 15.95 & 30.55 & --    & \multicolumn{1}{r}{--} \\
          & + Finetune & \textbf{58.88} & \textbf{54.71} & \textbf{70.68} & 10min10s & \multicolumn{1}{r}{1h36min45s} \\
          & + DREditor & 55.04 & 53.12 & 59.26 & \textbf{12.19s} & \multicolumn{1}{r}{\textbf{36.67s}} \\
\cmidrule{2-7}          & \multicolumn{1}{l|}{FiQA} & 8.62 & 5.92 & 11.69 & --    & \multicolumn{1}{r}{--} \\
          & + Finetune & \textbf{12.92} & \textbf{8.69} & \textbf{18.19} & 3h53min50s & \multicolumn{1}{r}{1d7h42min20S} \\
          & + DREditor & 10.18 & 7.16 & 13.78 & \textbf{13.28s} & \multicolumn{1}{r}{\textbf{6min24s}} \\
\cmidrule{2-7}          & \multicolumn{1}{l|}{NFCorpus} & 17.06 & 4.46 & 7.84 & --    & \multicolumn{1}{r}{--} \\
          & + Finetune & 13.65 & 3.43 & 5.25 & 1d5h57min20s & \multicolumn{1}{r}{7d23h9min16s} \\
          & + DREditor & \textbf{18.96} & 5.38 & \textbf{8.03} & \textbf{43.16s} & \multicolumn{1}{r}{\textbf{17min57s}} \\
    \midrule
    \multirow{15}[6]{*}{ANCE} & \multicolumn{1}{l|}{SciFact} & 51.72 & 46.59 & 65.96 & --    & \multicolumn{1}{r}{--} \\
          & + Finetune & 51.93 & 47.06 & \textbf{66.95} & 5min58s & \multicolumn{1}{r}{24min40s} \\
          & + LORA & \textbf{54.51} & \textbf{50.26} & 66.51 & 5min34s & \multicolumn{1}{r}{--} \\
          & + DREditor & 51.93 & 46.79 & 66.29 & \textbf{12.53S} & \multicolumn{1}{r}{\textbf{1min20s}} \\
\cmidrule{2-7}          & \multicolumn{1}{l|}{FiQA} & 27.09 & 20.81 & 33.23 & --    & \multicolumn{1}{r}{--} \\
          & + Finetune & 27.26 & 20.85 & 32.87 & 48min37s & \multicolumn{1}{r}{6h23min30s} \\
          & + LORA & \textbf{29.81} & \textbf{22.76} & \textbf{37.38} & 1h0min13s & \multicolumn{1}{r}{--} \\
          & + DREditor & 27.68 & 21.21 & 34.03 & \textbf{15.66s} & \multicolumn{1}{r}{\textbf{13min39s}} \\
\cmidrule{2-7}          & \multicolumn{1}{l|}{NFCorpus} & 22.01 & 7.30 & 10.94 & --    & \multicolumn{1}{r}{--} \\
          & + Finetune & 27.75 & 9.79 & 13.08 & 2h20min21s & \multicolumn{1}{r}{1d22h15min10s} \\
          & + LORA & \textbf{30.86} & \textbf{11.24} & \textbf{15.41} & 4h26min20s & \multicolumn{1}{r}{--} \\
          & + DREditor & 23.10 & 7.60 & 11.44 & \textbf{2min13s} & \multicolumn{1}{r}{\textbf{2min13s}} \\
    \bottomrule
    \end{tabular}%
    }
  \label{tab:MResultsExtrinsic}%
\end{table}%

From the results, it can be observed that our method still performs competitively across all three metrics. LoRA can achieve the same or even better performance than fine-tuning, but the time cost is still significant. These results are in line with previous studies that LoRA needs more training time than adapter finetuning, not to mention compared to DREditer.
This is because LoRA's primary savings are in computing memory. It fails to provide an improvement in the time efficiency of model training compared to fine-tuning. In contrast, our method still provides a substantial improvement in time efficiency.

\subsection{The Results of Calibrating $x_q$ Only}
\label{onesizebaby}
In theory, calibrating only $x_q$ is also possible to achieve a comparable performance as calibrating both $x_q$ and $x_a$. To find the difference between them, we show the results of calibrating $x_q$ alone in Table \ref{tab:One-side}. We have observed that making changes to only $x_q$ does not result in a significant difference compared to changing both $x_q$ and $x_a$ simultaneously. Editing both $x_q$ and $x_a$ yields slightly better results. However, we still recommend editing both $x_q$ and $x_a$ to ensure alignment with our optimization objective (Eq \ref{eq3}) and theoretically guarantee its effectiveness.

% Table generated by Excel2LaTeX from sheet 'Sheet11'
\begin{table}[!htb]
  \centering
  \caption{Results of calibrating $x_q$ only.}
  \resizebox{0.45\textwidth}{!}{
      \begin{tabular}{c|c|p{4.04em}|c|c}
    \toprule
    \textbf{Backbone} & \multicolumn{1}{c|}{\textbf{Method}} & \multicolumn{1}{c|}{\textbf{NDCG@10}} & \textbf{MAP@10} & \textbf{Recall@10} \\
    \midrule
    \multirow{9}[12]{*}{SBERT} & \multicolumn{4}{c}{SciFact} \\
\cmidrule{2-5}          & + DREditor$_{q}$ & \multicolumn{1}{c|}{70.18} & 67.35 & \textbf{77.94} \\
          & + DREditor$_{qa}$ & \multicolumn{1}{c|}{\textbf{72.07}} & \textbf{70.00} & 77.26 \\
\cmidrule{2-5}          & \multicolumn{4}{c}{FiQA} \\
\cmidrule{2-5}          & + DREditor$_{q}$ & \multicolumn{1}{c|}{\textbf{26.18}} & \textbf{19.96} & \textbf{32.36} \\
          & + DREditor$_{qa}$ & \multicolumn{1}{c|}{25.22} & 19.18 & 31.06 \\
\cmidrule{2-5}          & \multicolumn{4}{c}{NFCorpus} \\
\cmidrule{2-5}          & + DREditor$_{q}$ & \multicolumn{1}{c|}{30.30} & \textbf{10.73} & 14.20 \\
          & + DREditor$_{qa}$ & \multicolumn{1}{c|}{\textbf{30.37}} & 10.62 & \textbf{14.23} \\
    \midrule
    \multirow{9}[8]{*}{DPR} & \multicolumn{4}{c}{SciFact} \\
\cmidrule{2-5}          & + DREditor$_{q}$ & \multicolumn{1}{c|}{54.35} & 52.33 & \textbf{59.53} \\
          & + DREditor$_{qa}$ & \multicolumn{1}{c|}{\textbf{55.04}} & \textbf{53.12} & 59.26 \\
\cmidrule{2-5}          & \multicolumn{4}{c}{FiQA} \\
\cmidrule{2-5}            & + DREditor$_{q}$ & \multicolumn{1}{c|}{10.10} & 7.14  & 13.54 \\
          & + DREditor$_{qa}$ & \multicolumn{1}{c|}{\textbf{10.18}} & \textbf{7.16} & \textbf{13.78} \\
\cmidrule{2-5}          & \multicolumn{4}{c}{NFCorpus} \\
\cmidrule{2-5}            & + DREditor$_{q}$ & \multicolumn{1}{c|}{18.92} & \textbf{5.39} & 7.99 \\
          & + DREditor$_{qa}$ & \multicolumn{1}{c|}{\textbf{18.96}} & 5.38  & \textbf{8.03} \\
    \midrule
    \multirow{9}[12]{*}{ANCE} & \multicolumn{4}{c}{SciFact} \\
\cmidrule{2-5}          & + DREditor$_{q}$ & \multicolumn{1}{c|}{51.88} & \textbf{47.36} & 64.12 \\
          & + DREditor$_{qa}$ & \multicolumn{1}{c|}{\textbf{51.93}} & 46.79 & \textbf{66.29} \\
\cmidrule{2-5}          & \multicolumn{4}{c}{FiQA} \\
\cmidrule{2-5}          & + DREditor$_{q}$ & \multicolumn{1}{c|}{\textbf{28.41}} & \textbf{21.98} & \textbf{34.67} \\
          & + DREditor$_{qa}$ & \multicolumn{1}{c|}{27.68} & 21.21 & 34.03 \\
\cmidrule{2-5}          & \multicolumn{4}{c}{NFCorpus} \\
\cmidrule{2-5}          & + DREditor$_{q}$ & \multicolumn{1}{c|}{\textbf{23.44}} & \textbf{7.96} & \textbf{11.54} \\
          & + DREditor$_{qa}$ & \multicolumn{1}{c|}{23.10} & 7.60  & 11.44 \\
    \bottomrule
    \end{tabular}%
    }
  \label{tab:One-side}%
\end{table}%

\subsection{More Results of Visualization of the Embedding Projections}
\label{morereulsts}
This section presents additional results on the FiQA and NFCorpus datasets to further illustrate the impact of DREditor on the embeddings of questions and answers, as previously discussed in Section \ref{intrin eva}. Figures \ref{fig:fiqa-sbert} and \ref{fig:nfcorpus-sbert} depict the embedding of test data before and after applying DREditor with backbone SBERT on FiQA and NFCorpus, respectively.

The conclusion drawn in Section \ref{intrin eva} remains unchanged. On the one hand, the DREditor calibrates the embedding of question $x_q$ to the new embedding $x^{'}_q$ that is very similar to the embedding of answers $x_a$. On the other hand, it also shows that the effect of the calibration on the target domain (training) data and the source domain data is different. Interestingly, calibrating the source domain data on these two datasets makes the embedding of questions closer to the answers than the dataset of SciFact. This may indicate that the constructed structured knowledge graph from WikiData and the unstructured text facts from ChatGPT on these datasets are more interconnected than those on SciFact. The reason for this is that the answers in SciFact come from the newest paper, which may not be covered by WikiData and ChatGPT.

\begin{figure*}[!htb]
    \centering
    \includegraphics[width=0.95\textwidth]{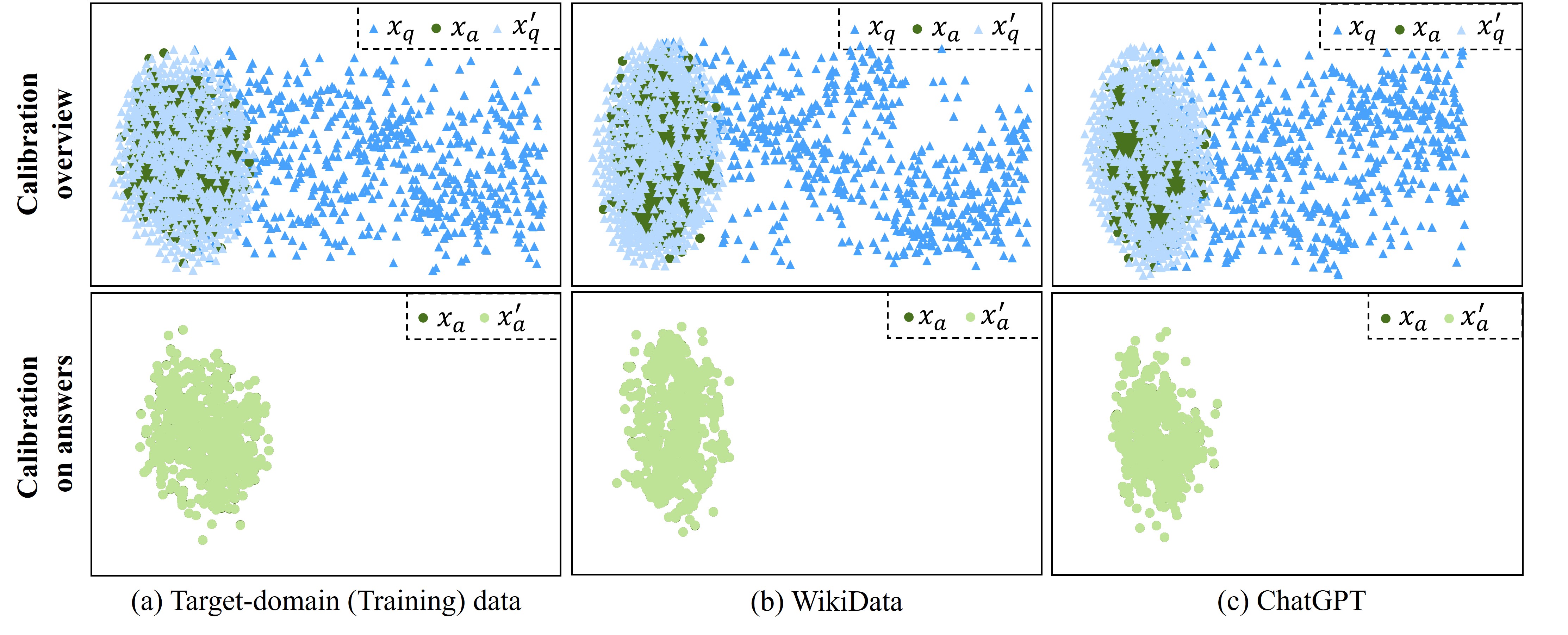}
    \caption{The embedding space on FiQA datasets with backbone SBERT}
    \label{fig:fiqa-sbert}
     \vspace{-3mm}
\end{figure*}

\begin{figure*}[!htb]
    \centering
    \includegraphics[width=0.95\textwidth]{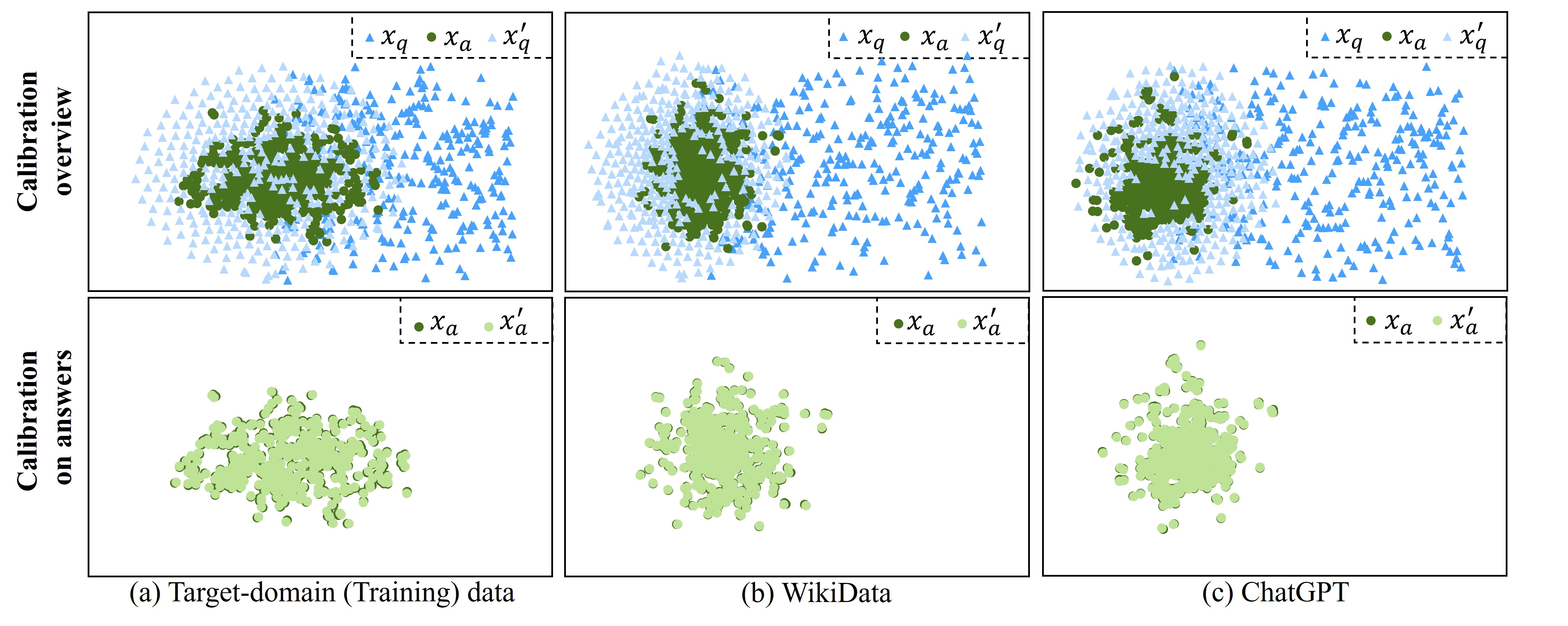}
    \caption{The embedding space on NFCorpus datasets with backbone SBERT}
    \label{fig:nfcorpus-sbert}
    \vspace{-3mm}
\end{figure*}

\section{Details of Datasets and Baselines}
\label{imple}
\label{dataset}
We resort to the datasets and baselines on BEIR \cite{thakur2021beir} to evaluate how well DREditor performs on tasks that are specialized to different domains, including the zero-shot retrieval scenario.

\begin{table*}[!htb]
\centering
\caption{Examples of generated question-answer pair based on WikiData and ChatGPT.}
\label{sec:appendixd}
\resizebox{0.95\textwidth}{!}{%
\begin{tabular}{l|l}
\toprule
\multicolumn{2}{c}{\textbf{Structured knowledge graph from WikiData}} \\  \midrule
\textbf{Knowledge} & \begin{tabular}[c]{@{}l@{}}Triple: (Biomaterial, has quality, Biocompatibility)\\ Masked Triple: (What, has quality, Biocompatibility)\\ (\textit{Entity from TagMe, relation from WikiData})\end{tabular} \\  \midrule
\textbf{Question} & \begin{tabular}[c]{@{}l@{}}What possesses biocompatibility of high quality\\ (Paraphrased by ChatGPT)\end{tabular} \\  \midrule
\textbf{Answer} & \begin{tabular}[c]{@{}l@{}}Any substance that has been engineered to interact with biological systems for a medical purpose\\ (\textit{Description information of Biomaterial, extracted from WikiData})\end{tabular} \\  \midrule
\multicolumn{2}{c}{\textbf{Unstructured text facts from ChatGPT}} \\  \midrule
\textbf{Knowledge} & \begin{tabular}[c]{@{}l@{}}0-dimensional biomaterials show inductive properties\\ (\textit{Sample from test data})\end{tabular} \\  \midrule
\textbf{Question} & \begin{tabular}[c]{@{}l@{}}Factual info. and explanations (\textit{Generated by ChatGPT}): \\0-dimensional biomaterials refer to materials that are extremely small, such as nanoparticles. \\ These materials have been found to have inductive properties, \\  \\
Question (\textit{Generated by ChatGPT}):\\
which means they can stimulate the growth and differentiation of cells and tissues. \end{tabular} \\  \midrule
\textbf{Answer} & \begin{tabular}[c]{@{}l@{}}What are 0-dimensional biomaterials, and how do they stimulate the growth and differentiation of cells and tissues?\\ (\textit{Generated by ChatGPT})\end{tabular} 
\\ \bottomrule
\end{tabular}%
}
\end{table*}

\subsection{Details of Datasets} 
Note that some datasets on BEIR are not publicly available or do not have training data, we omit those datasets on BEIR in our experiments. Consequently, we utilize datasets from three domains, i.e., financial, bio-medical, and scientific domains. Specifically, we use FiQA-2018 \cite{10.1145/3184558.3192301} to conduct the financial evaluation, which is an opinion-based question-answering task. NFCorpus \cite{boteva2016full}, harvested from NutritionFacts, is used for the bio-medical information retrieval task. SciFact \cite{wadden-etal-2020-fact} is used for science-related fact-checking, which aims to verify scientific claims using evidence from the research literature. The statistics of these three datasets are shown in Table \ref{tab:dataset}.
Meanwhile, datasets on BEIR have been well-organized in the form of query-corpus pairs, which we regard as question-answer pairs. In our experiments, we only use test data for evaluation without model training.

\begin{table}[!htb]
  \centering
  \caption{The statistic of three datasets.}
    \resizebox{0.35\textwidth}{!}{
    \begin{tabular}{c|c|c|c}
    \toprule
    \textbf{Datasets} & \textbf{Dataset type} & \textbf{corpus} & \textbf{query} \\
    \midrule
    \multirow{3}[6]{*}{SciFact} & Train & \multirow{3}[6]{*}{5183} & 809 \\
\cmidrule{2-2}\cmidrule{4-4}          & Dev   &       & - \\
\cmidrule{2-2}\cmidrule{4-4}          & Test  &       & 300 \\
    \midrule
    \multirow{3}[6]{*}{FiQA} & Train & \multirow{3}[6]{*}{57638} & 5500 \\
\cmidrule{2-2}\cmidrule{4-4}          & Dev   &       & 500 \\
\cmidrule{2-2}\cmidrule{4-4}          & Test  &       & 648 \\
    \midrule
    \multirow{3}[6]{*}{NFCorpus} & Train & \multirow{3}[6]{*}{3633} & 2590 \\
\cmidrule{2-2}\cmidrule{4-4}          & Dev   &       & 324 \\
\cmidrule{2-2}\cmidrule{4-4}          & Test  &       & 323 \\
    \bottomrule
    \end{tabular}
    }
  \label{tab:dataset}%
\end{table}%

\subsection{Baseline Implementation} %\todo{DPR}
We consider the representative dense-retrieval-based baselines on BEIR, including SBERT \cite{reimers-2019-sentence-bert}, DPR \cite{karpukhin-etal-2020-dense}, and ANCE \cite{xiongapproximate}.
Each baseline is initialized by the public available checkpoint that is pre-trained based on an open-domain corpus, such as MS MARCO\footnote{\url{https://microsoft.github.io/MSMARCO-Question-Answering/}} extracted from the Web and Natural Questions\cite{kwiatkowski-etal-2019-natural} mined from Google Search.
We then fine-tune each baseline on the training data of the corresponding dataset and evaluate the performance of DREditor on the fine-tuned baselines. Meanwhile, to reduce the computation and ensure the fairness of the comparison, we add a layer to the original baseline model and train this layer only.
We utilize their raw dev set for model validation during the fine-tuning.
As for SciFact, it has no dev set, and the size of the training dataset is small. Thus, we split the raw training data into two folds under the division ratio of $1:1$, i.e., the training data for fine-tuning, and the dev set for validation, respectively.
We use \textit{MultipleNagetiveRankingLoss}, implemented in BEIR, to finetune SBERT and ANCE. Specifically, we fine-tune them for 10 epochs using AdamW optimizer, together with a learning rate of $2e-5$ and batch size of $16$.
%We optimize each baseline using the AdamW optimizer and fine-tune it for 10 epochs with a learning rate of $2e-5$. The batch size is $16$ for SBERT and ANCE.
Regarding DPR, we also fine-tune it for 10 epochs using the AdamW optimizer, together with a learning rate of $2e-5$. Due to DPR requiring a larger GPU memory space, we had to reduce its batch size to $8$. Also, we fine-tuned DPR by contrastive loss used in the original paper \cite{karpukhin-etal-2020-dense}, and we randomly selected three negative data for each positive data. As for the SOTA ZeroDR model COCO-DR \cite{yu-etal-2022-coco}, we utilize the open-sourced checkpoint\footnote{\url{https://huggingface.co/OpenMatch/cocodr-base-msmarco}} in our ZeroDR experiments.

\begin{figure*}[t]
    \centering
    \includegraphics[width=0.9\textwidth]{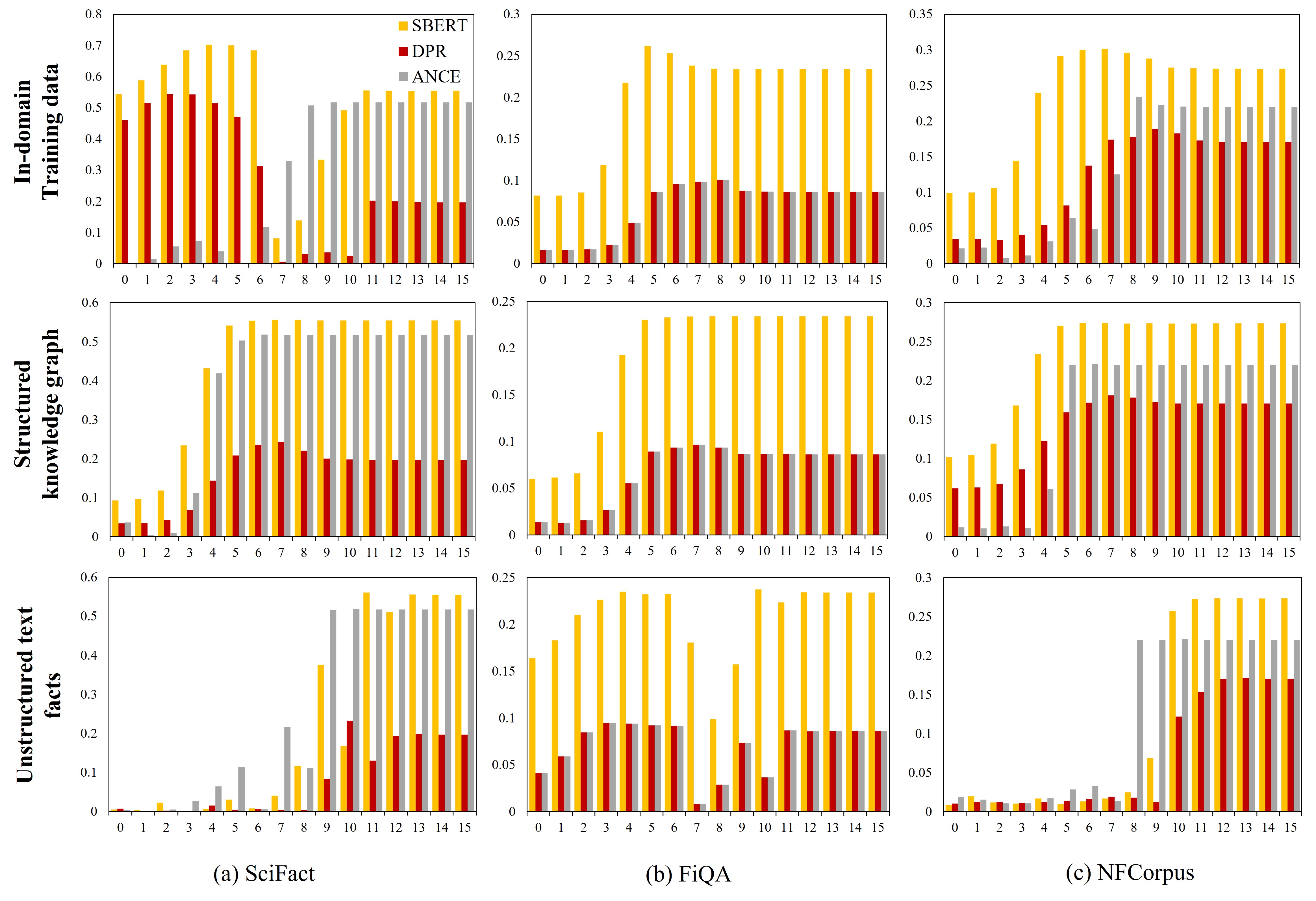}
    \caption{Analysis on Hyper-parameters}
    \label{fig:HYPER}
    \vspace{-3mm}
\end{figure*}

\section{Details on source-domain data representation}
\label{456789}
As discussed in Section \ref{7uj2w3o8e34rweoh}, we consider the structured knowledge graph extracted from WikiData and the unstructured text facts from ChatGPT as the source-domain data for the zero-shot scenario. Here, we show how to transform them into question-answer pairs. Refer to Table \ref{sec:appendixd} for examples.

\begin{itemize}[leftmargin=*]
    \item \underline{Structured knowledge graph from WikiData}. In our experiments, our tasks cover the fields of finance, bio-medicine, and science. We use WikiData to construct a background knowledge graph for each task. In particular, we first resort to TagMe \cite{ferragina2010tagme} to extract entities $h$ and then crawl from WikiData to obtain triplet $(h, r, t)$ and description information of $h$ and $t$. To transform them into question-answer pairs, we mask out head entity $h$ (or tail entity $t$) in a triple with the word "what" and concatenate the triple element together as a sentence $[S]$, which takes the form of "what $r$ $t$" or "$h$ $r$ what". For example, given a fact triple (\textit{Biomaterial}, \textit{has quality}, \textit{Biocompatibility}), we mask out the head entity and then yield "\textit{what has quality Biocompatibility}". Next, ChatGPT is utilized to paraphrase the sentence. Specifically, we launch the following question to ChatGPT "\textit{Paraphrase the following sentence and keep the keywords unchanged: [S]}", and obtain the results as the question. Finally, the corresponding answer is the description information of the masked entity from WikiData. Namely, given $[S]$="\textit{what has quality Biocompatibility}", ChatGPT receives a question: "\textit{Paraphrase the following sentence and keep the keywords unchanged: what has quality Biocompatibility}", and it returns a question: "\textit{What possesses biocompatibility of high quality}".
    \item \underline{Unstructured text facts from ChatGPT}. It is known that large pre-trained language models store factual knowledge present in the pre-training data \cite{petroni2019language, wang2020language}. In our experiments, we construct an open-domain background knowledge for each task using ChatGPT. In particular, given a sample text $[T]$ from test data (e.g., a query or a corpus), we require ChatGPT to offer factual information and explanations based on a template "\textit{Give me the factual information and explanations about the following statement or question: [T]}". The response, denoted as $a$, is further used to generate a question $q$ based on another template "\textit{Read the following factual information and generate a question: 1) $a$}". In this way, a question-answer pair that implies the factual knowledge of ChatGPT is obtained.
\end{itemize}

\begin{table}[!htb]
  \centering
  \caption{The statistics of the source-domain data.}
\resizebox{0.35\textwidth}{!}{%
    \begin{tabular}{c|c|c|c}
    \toprule
    \textbf{Datasets} & \textbf{Data type} & \textbf{KG}    & \textbf{ChatGPT} \\
    \midrule
    \multirow{2}[4]{*}{SciFact} & question & 5951  & 300 \\
\cmidrule{2-4}          & answer & 26369 & 300 \\
    \midrule
    \multirow{2}[4]{*}{FiQA} & question & 6708  & 647 \\
\cmidrule{2-4}          & answer & 26694 & 647 \\
    \midrule
    \multirow{2}[4]{*}{NFCorpus} & question & 3511  & 323 \\
\cmidrule{2-4}          & answer & 15950 & 323 \\
    \bottomrule
    \end{tabular}%
}
  \label{tab:number}%
   \vspace{-3mm}
\end{table}%

In practice, as ChatGPT requires a fee and has been banned in some countries, we use Monica\footnote{\url{https://monica.im/}} for annotations, which is a Chrome extension powered by ChatGPT API. As free users of Monica, we have a daily usage limit. Thus, we require Monica to annotate multiple instances at a time (precisely, 20 at a time). 
Here, we present the specific number of the \underline{Structured Knowledge Graph} from WikiData (denoted as KG) and \underline{Unstructured Text Facts} from ChatGPT (denoted as ChatGPT) constructed on each dataset, as shown in Table \ref{tab:number}.

\section{Analysis on Hyper-parameters}
\label{hyper}

This section examines the properties of the hyper-parameter $\lambda = \hat\beta^2 n$ in $A = \lambda \sum_{i=1}^n \frac{1}{n} x^i_a {x^i_a}^{T}$ of Eq (\ref{eq4}).
While we choose $\lambda$ based on the dev set performance in our main experiments, we reveal some interesting observations related to it. ($n$ is adjusted by the ratio of the dev set and train set).
We assess the impact of varying $\lambda$ values on performance across different datasets and tasks, as illustrated in Figure \ref{fig:HYPER}. Here, the x-axis in Figure \ref{fig:HYPER} represents the power of 10 for $\lambda$, while the y-axis represents the value of nDCG@10.

Increasing $\lambda$ initially improves performance until it reaches a peak and then decreases to a stable level. This is because a small $\lambda$ value reduces the weight of the second term in Eq (\ref{eq3}) and Eq (\ref{finaleq}), causing significant changes to $x_a$ and resulting in incorrect retrieval outcomes. Conversely, a large $\lambda$ value preserves the original relationship and avoids editing the matching rules that should be changed, thereby maintaining the original performance value.

When using unstructured text facts for editing, a small $\lambda$ value leads to poor effectiveness, which is an extreme scenario due to the significant lack of relevance between unstructured text facts and the specific domain. Consequently, the small $\lambda$ cannot maintain the embedding of the answers $x_a$ unaltered. Therefore, only when $\lambda$ becomes sufficiently large can \textit{DREitor} preserve the original relationship as much as possible and extract useful information at the same time.

\end{document}